\documentclass[amsthm]{elsart}

\usepackage{yjsco}
\usepackage{natbib}

\usepackage{amssymb, amsbsy, amsfonts}
\usepackage{amsmath, amsthm}
\usepackage{latexsym,float,color}

\allowdisplaybreaks[4]

\newtheorem{lemma}{Lemma}
\newtheorem{theorem}{Theorem}
\newtheorem{corollary}{Corollary}

\newtheorem{remark}{Remark}
\newtheorem{example}{Example}
\newtheorem{assumption}{Assumption}

\newenvironment{ProblemSpec}[1]{\noindent#1\\}{}
\newcommand{\SigmaP}{\texttt{Sigma}}

\newcommand{\ProblemRS}{\textsf{RS}}
\newcommand{\FLSR}{\textsf{FLSR}}

\newcounter{mmacnt}
\def\restartmma{\setcounter{mmacnt}{0}}
\restartmma \catcode`|=\active
\def|#1|{\mathrm{#1}}
\catcode`|=12
\newenvironment{mma}{
 \par\smallskip
 \catcode`|=\active
 \parskip=0pt\parindent=0pt 
 \small
 \def\In##1\\{%
   \def\linebreak{\hfill\break\null\qquad}%
   \refstepcounter{mmacnt}
   \hangindent=2.5em\hangafter=0
   \leavevmode
   \llap{\tiny\sffamily In[\arabic{mmacnt}]:=\kern.5em}%
   \mathversion{bold}\footnotesize$\displaystyle##1$\normalsize
   \mathversion{normal}\par
 }%
 \def\Print##1\\{%
   \def\linebreak{\hfill\break}%
   \hangindent=2.5em\hangafter=0
   \leavevmode ##1\par}%
 \def\Out##1\\{%
   \def\linebreak{$\hfill\break\null\hfill$}%
   \kern\abovedisplayskip\par
   \hangindent=2.5em\hangafter=0
   \leavevmode
   \llap{\tiny\sffamily Out[\arabic{mmacnt}]=\kern.5em}
   \footnotesize$\displaystyle##1$\normalsize\hfill\null\par
   \kern\belowdisplayskip
 }%
 \def\Warning##1##2\\{%
   \def\linebreak{\hfill\break}%
   \hangindent=2.5em\hangafter=0
   \leavevmode
   {\scriptsize##1 : ##2}\par}%
}{%
 \par\smallskip
}

\let\set\mathbb

\newcommand{\ep}{\varepsilon}

\begin{document}

\begin{frontmatter}

\title{A Symbolic Summation Approach to\\ Feynman Integral Calculus}

\thanks{Supported by the Austrian Science Fund
   (FWF) grants P20162-N18, P20347-N18, DK W1214, the German Research Fund SFB-TR-9, the EU grants TMR
   Network Heptools and PITN-GA-2010-264564.}

\author{Johannes Bl\"umlein}
\address{Deutsches Elektronen-Synchrotron, DESY, Platanenallee 6, D-15738 Zeuthen,
Germany}
\ead{Johannes.Bluemlein@desy.de}

 \author{Sebastian Klein}
 \address{Institut f\"ur Theoretische Teilchenphysik und Kosmologie,
                 RWTH Aachen University, \\ D--52056 Aachen, Germany}
 \ead{sklein@physik.rwth-aachen.de}

\author{Carsten Schneider}
\address{Research Institute for Symbolic Computation, RISC, Johannes
Kepler University Linz, Austria}
\ead{cschneider@risc.jku.at}

\author{Flavia Stan}
\address{Research Institute for Symbolic Computation, RISC, Johannes
Kepler University Linz, Austria}
\ead{fstan@risc.jku.at}

\begin{abstract}
\noindent
Given a Feynman parameter integral, depending on a single discrete variable $N$ and a real
parameter $\varepsilon$, we discuss a new algorithmic framework to compute the first coefficients
of its Laurent series expansion in $\varepsilon$. In a first step, the integrals are expressed by
hypergeometric multi-sums by means of symbolic transformations. Given this sum format,
we develop new summation tools to extract the first coefficients of its series expansion whenever they are expressible in terms of indefinite nested product-sum expressions. In
particular, we enhance the known multi-sum algorithms to derive recurrences for sums with
complicated boundary conditions, and we present new algorithms to find formal Laurent series
solutions of a given recurrence relation.
\end{abstract}


\end{frontmatter}

\section{Introduction}

Starting with single summation over hypergeometric terms developed, e.g., in~\cite{Gosper,Zeilberger-identities,Hyper,Abramov:94,GFF} symbolic summation has been intensively enhanced to multi-summation like, e.g., the holonomic approach of~\cite{Zeilberger-holonomic,Chyzak,Schneider:05,Chris-thesis}. In this article we use the techniques of~\cite{Celine,WZtheory} which lead to efficient algorithms developed, e.g., in~\cite{Wegschaider} to compute recurrence relations for hypergeometric multi-sums. Besides this, we rely on multi-summation algorithms presented in~\cite{Sigma-Manual} that generalize the summation techniques worked out in~\cite{AequalB}; the underlying algorithms are based on a refined difference field theory elaborated in~\cite{Schneider:08c,Schneider:11} that is adapted from Karr's $\Pi\Sigma$-fields originally introduced in~\cite{Karr}.

We aim at combining these summation approaches which leads to a new framework to evaluate Feynman integrals. In a nutshell, given a Feynman integral, we transform it to hypergeometric multisums, compute afterwards linear recurrences for these multi-sums, and finally decide constructively by recurrence solving whether the integrals (resp.\ the multisums) have series expansions whose coefficients can be represented in terms of indefinite nested sums and products.
The method consists of a completely algebraic
algorithm. It is therefore well-suited for implementation in computer algebra
systems.

We show in a first step that Feynman parameter
integrals, which contain local operator insertions, in $D$-dimensional Minkowski space with one time- and $(D-1)$
Euclidean space dimensions, $\varepsilon = D - 4$ and $\varepsilon \in {\mathbb R}$ with
$|\varepsilon| \ll 1$, can be transformed by means of symbolic computation to hypergeometric multi-sums ${\cal S}(\ep,N)$ with $N$ an integer parameter. Given this representation,
one can check by analytic arguments whether the integrals can be expanded in a Laurent series w.r.t.\ the parameter $\ep$, and we seek summation algorithms to compute the first few coefficients of this expansion whenever they are representable in terms of indefinite nested sums and products. Due to the special input
class of Feynman integrals, these solutions can be usually transformed to harmonic
sums or $S$-sums; see~\cite{BK,Vermaseren:99,Moch:02,Ablinger}.

In general, we present an algorithm (see Theorem~\ref{Thm:AlgForMultiSum}) that decides constructively, if these first coefficients of the $\ep$--expansion can be written in such indefinite nested product-sum expressions.
Here one first computes a homogeneous recurrence by WZ-theory and Wegschaider's approach. This recurrence together with initial values gives an alternative representation for the series expansion (see Lemma~\ref{Lemma:UniqueSeqSol}). Moreover, we develop a recurrence solver (see Corollary~\ref{Cor:ExpansionSolver}) which computes the coefficients of the expansion in terms of indefinite nested product-sum expressions whenever this is possible. The backbone of this solver relies on algorithms from~\cite{Hyper,Abramov:94,Schneider:01,Schneider-solving}. Since the solutions are highly nested by construction, their simplification to sum representations with minimal depth are crucial; see~\cite{Schneider:11}.

From the practical point of view there is one crucial drawback of the proposed solution: looking for such recurrences is extremely expensive.
For our examples arising from particle physics the proposed algorithm is not applicable considering
the available computer and time resources.
On that score we relax this very restrictive requirement and search for possibly inhomogeneous recurrence relations. However, the input sums have summands which present poles outside
the given summation ranges. Combining Wegschaider's package \texttt{MultiSum} and the new package \texttt{FSums} presented in~\cite{Stan-thesis}
we determine recurrences with inhomogeneous sides consisting of well-defined sums with fewer sum quantifiers. Applying our method to these simpler sums by recursion will lead to an expansion of the right hand side of the starting recurrence. Finally, we compute the coefficients of the original input sum by our new recurrence solver mentioned above.

The outline of the article is as follows. In Section~\ref{Sec:Physic} we explain all computation steps that lead
from Feynman integrals to hypergeometric multi-sums of the form~\eqref{eq:A4} which can be expanded in a Laurent
expansion~\eqref{Equ:LaurentExp} where the coefficients $F_i(N)$ can be represented in the form~\eqref{eq:A5}. In
the beginning of Section~\ref{Equ:FirstApproach} we face the problem that the multi-sums~\eqref{eq:A4} have to be
split further in the form~\eqref{physics-sums-general} to fit the input class of our summation algorithms. We first
discuss convergent sums only. The treatment of
those sums which diverge in this special format or sums with several infinite summations that have difficult
convergence properties will be dealt with later, cf. Remark~\ref{Remark:MethodConclusion}. In the remaining parts
of
Section~\ref{Equ:FirstApproach} we present the general mechanisms to compute the first coefficients $F_i(N)$ for a given hypergeometric multi-sum.
In Section~\ref{Sec:FindRecurrences} we present an algorithmic approach to hypergeometric sums with non-standard boundary conditions. This allows us to generate the inhomogeneous sides of recurrences delivered by Wegschaider's package \texttt{MultiSum}.
Finally, in Section~\ref{Sec:EfficientApproach} we obtain a method that is capable of computing the coefficients $F_i(N)$ in reasonable time. Conclusions are given in Section~\ref{Sec:Conclusion}.

\section{Multiple sum representations of Feynman integrals}\label{Sec:Physic}

\noindent We show how integrals
emerging in renormalizable Quantum Field Theories, like
Quantum Electrodynamics or Quantum Chromodynamics, see e.g. \cite{ANCONT2},
can be transformed by means of symbolic computation to hypergeometric multi-sums.
We study a very general class of Feynman integrals which are of relevance for many
physical processes at high energy colliders, such as the Large Hadron Collider, LHC, and others.

The processes obey special-relativistic kinematics with energy-momentum vectors in Minkowski space,
$\mathbb{M}^D$, see e.g.,~\cite{Naas}, i.e., a $D$-dimensional linear space where
the elements $a=(a_0,\vec{a})\in\mathbb{M}^D$ decompose into the time coordinate $a_0\in\mathbb{R}$ and the spatial coordinates $\vec{a}\in\mathbb{R}^{D-1}$ which form a $D-1$-dimensional Euclidean subspace; the bilinear form is defined by $a.b \equiv \langle a,b \rangle=a_0b_0-\vec{a}\vec{b}\in \mathbb{R}$ for $b=(b_0,\vec{b}) \in \mathbb{M}^D$. Below analytic continuations in $D:=4+\ep$ with $\ep\in\set R$ are considered. Here we study integrals

\vspace*{-0.2cm}

\begin{equation}
{\cal I}(\varepsilon,N,p) = \int \frac{d^D p_1}{(2\pi)^D} \ldots \int \frac{d^D
p_k}{(2\pi)^D}
\frac{{\cal N}(p_1, \ldots p_k; p; M^2;
\Delta, N)}{(-p_1^2 + m_1^2)^{l_1} \ldots (-p_k^2 + m_k^2)^{l_k}}
\prod_V \delta_V
\label{eq:A7}
\end{equation}
with $\Delta, p, p_i \in \mathbb{M}^D$ and $m_i \in \{0,M\}$ for some $M\in\set R$ with $M>0$. The restriction that there is only one mass $M$ is the only one specifying the class of Feynman diagrams from arbitrary ones.
The propagator powers $l_i$ obey $l_i \in \mathbb{N}$ and for the special vector $\Delta$ in (\ref{eq:A7}) one has
$\Delta.\Delta = 0$. The numerator
${\cal N}$ is usually given in terms of finite sums where the range depends on a discrete parameter $N$ and where the summand depends on the scalars $p.p_j, p_i.p_j,
\Delta.p_i$ ($1\leq i,j\leq k$), on $M^2$ and on $N$. In particular, for each $N\in\set N$, ${\cal N}$ is a polynomial in terms of these scalars and $M^2$ where the exponents of the $\Delta.p_j$ ($1\leq j\leq k$) in a given monomial sum up to $N$ and the exponents of the remaining scalars and $M^2$ are constant.
The $\delta_V$ occurring in~\eqref{eq:A7} are shortcuts for Dirac delta functions  in $D$ dimensions  $\delta_V \equiv
\delta^{(D)}\left(\sum_{l=1}^k a_{V,l}
p_l\right), a_{V,l} \in \mathbb Q$. I.e., if $a_{V,i}\neq0$, we get
\vspace*{-0.3cm}

\begin{equation}
\int d^D p_i \delta^{(D)}\left(\sum_{l=1}^k a_{V,l} p_l\right)f(p_i):= \frac{f(p_i)}{|a_{V,i}|}\Big|_{p_i=u}\quad\text{with}\quad u:=-\frac{1}{a_{V,i}}\sum_{l=1,l\neq i}^k a_{V,l}p_l;
\label{eq:del}
\end{equation}
here $f$ stands for the integrand of~\eqref{eq:A7}. For each such rule~(\ref{eq:del}) for the remaining $\delta_V$, one integral sign in (\ref{eq:A7}) can be eliminated. As a consequence we obtain integrals of the same shape but with fewer integral signs.
Such an integral
may be easily linearly transformed
into Euclidean integrals (Wick rotation, \cite{WICK1,WICK2}) in the
Euclidean space by replacing $a=(a_0,\vec{a})\in\mathbb{M}^D$ with $\bar{a}=(i a_0,\vec{a})$. In this way, for $b=(b_0,\vec{b})$ the bilinear form $\langle\bar{a},\bar{b}\rangle=
-a_{0}b_{0}
-\vec{a}. \vec{b} < 0$  obtains a definite sign; $\sqrt{-\langle \bar{a},\bar{a} \rangle}$ is then the
Euclidean norm $||\bar{a}||$. Summarizing, we obtain an Euclidean integral of the same shape as~\eqref{eq:A7} with the Euclidean momenta $\bar{p}_{i},\bar{p}$ (instead of $p_i,p$) and where the denominators can be written in the form $((\sum_{j=1}^kc_j^{(i)}\bar{p}_{j})^2+m_i^2)^{l_i}$ with $c_j^{(i)}\in\mathbb Q$ (instead of $(-p_i^2+m_i^2)^{l_i}$); this format is due to the usage of~\eqref{eq:del}.\\
Subsequently, we show how this Euclidean integral can be mapped to an integral on an $m$-dimensional unit cube. Define $D_i:=(\sum_{j=1}^kc_j^{(i)}\bar{p}_{j})^2+m_i^2$. Then we loop over $r$ ($r=1,2,\dots,k$) as follows. For the $r$th iteration, fix  $\bar{q}:=\bar{p}_{r}$. W.l.o.g.\ assume that $c_r^{(i)}\in\{0,1\}$ for $1\leq i\leq k$. Now collect those denominator factors $D_i^{l_i}$ where $\bar{q}$ occurs, say $\prod_{j=1}^n D_{i_j}^{l_{i_j}}$ ($n\in\set N$). Then we use the formula

\vspace*{-0.5cm}

\begin{equation}
\label{eq:FEYN1}
\frac{1}{\prod_{j=1}^n D_{i_j}^{l_{i_j}}} = \tfrac{\Gamma(l)}{\prod_{j=1}^n\Gamma(l_{i_j})}
\int_0^1 d x_1 \ldots \int_0^1 d x_n \delta\bigg(\sum_{j=1}^n x_j - 1\bigg)
\frac{\prod_{j=1}^n x_j^{{l_{i_j}}-1}}{(x_1 D_{i_1} + \ldots x_{n} D_{i_n})^{l}}
\end{equation}
with $l=\sum_{j=1}^n l_{i_j}$; here $\delta$ is the Dirac delta function, the variables $x_k$ are called Feynman parameters, and
$\Gamma(z)$ denotes the Gamma-function.
Due to the Dirac delta function, we get that $A:=x_1 D_{i_1} + \ldots x_{n} D_{i_j}=
\bar{q}^2+a.\bar{q}+b$ where $a$ and $b$ are expressions free of $\bar{q}$. Hence we can write $A=(\bar{q}+a/2)^2+R$ with $R:=-a^2/4+b$ being free of $\bar{q}$.
Replacing the denominator of our integral by this formula, we can simplify $A$ further.
Namely, using the shift-invariance w.r.t. the vector $\bar{q}$,
which holds in $D$-dimensional Euclidean space, the denominator $A$ can be brought to the form $(\bar{q}^2+R)$ without changing the integral.
Finally, expanding the numerators and applying the $\bar{q}$-integral termwise lead to integrals of the form
$\int \frac{d^D \bar{q}}{(2 \pi)^D}\frac{\prod_{\lambda=1}^m q_{\lambda}.\bar{q}}{(\bar{q}^2+R)^l}$ where the expression $q_{\lambda}$ is free of $\bar{q}$. If $m$ is odd, i.e., an odd number of vector multiplications w.r.t.\ $\bar{q}$ arise, the integral evaluates to $0$ by symmetry. If $m$ is even, one exploits the simplification

\vspace*{-0.2cm}

$$\int \frac{d^D \bar{q}}{(2 \pi)^D}\frac{\prod_{\lambda=1}^{m/2} q_{\lambda}.\bar{q}}{(\bar{q}^2+R)^l}= r(D) \int \frac{d^D k}{(2 \pi)^D} \frac{(\bar{q}^2)^r}{(\bar{q}^2+R)^l}$$ where $r(D)$ stands for a rational function in $D$ (i.e., in $\ep$) that can be determined by an explicit formula. To this end, the following formula is applied to the remaining integrals:
$$\int \frac{d^D \bar{q}}{(2 \pi)^D} \frac{(\bar{q}^2)^r}{(\bar{q}^2+R)^l} = \frac{1}{(16 \pi^2)^{D/4}} \frac{\Gamma(r+D/2)
\Gamma(l-r-D/2)}{\Gamma(D/2) \Gamma(l) (R^2)^{l-r-D/2}}.$$
Usually, these operations are carried out in terms of tensors to keep the size compact and to determine additional relations efficiently.
The above procedure is repeated until all momentum
integrals for the $p_r$ ($r=1,2,\dots,k$) are computed. As a result one is left with the integrals
over $x_i \in [0,1]$, equipped with a pre-factor $C(\varepsilon, N,M)$.

\vspace*{2mm}\noindent
{{\sf Step~1: From Feynman parameter integrals to Mellin--Barnes integrals
and multinomial series.}}

\vspace*{1mm} \noindent
Parts of these scalar integrals again can be
computed trivially related to the $\delta$-distributions,

\vspace*{-0.2cm}

\begin{equation*}
\int_0^1 dx_l \delta\bigg(\sum_{k=1}^n x_k - 1\bigg)
= \theta \bigg(1 - \sum_{k=1, k \neq l}^n x_k \bigg) \prod_{m=1, m \neq l}^n
\theta(x_m), \end{equation*}
where $\theta(z)$ is 1 if $z\geq0$ and 0 otherwise.
There may be more integrals, which can be computed,
usually as indefinite integrals, without special effort.
Mapping all Feynman-parameter integrals onto the $m$-dimensional unit cube (as described above) one obtains the following structure~:
\begin{equation}
\label{eq:A9}
{\cal I}(\ep,N) = C(\ep, N, M) \int_0^1 dy_1 \ldots \int_0^1 dy_m
\frac{\sum_{i=1}^k \prod_{{l}=1}^{r_i}
[P_{i,l}(y)]^{\alpha_{i,l}(\varepsilon,N)}}{[Q(y)]^{\beta(\varepsilon)}}~,
\end{equation}
with $k\in\set N$, $r_1,\dots,r_k\in\set N$ and
where $\beta(\ep)$ is  given by a rational function in $\ep$, i.e., $\beta(\ep)\in\set Q(\ep)$, and similarly
$\alpha_{i,l}(\ep,N) = n_{i,l} N + \overline{\alpha}_{i,l}$ for some $n_{i,l} \in \{0,1\}$ and $\overline{\alpha}_{i,l}\in\set Q(\ep)$, see also \cite{BOGNER}
in the case when no local operator insertions are present.
$C(\ep, N, M)$ is a factor which depends on the dimensional parameter $\ep$,
the integer parameter $N$ and $M$.
$P_i(y), Q(y)$ are polynomials in the remaining Feynman parameters $y=(y_1,\dots,y_m)$ written in multi-index notation.
In (\ref{eq:A9}) all terms
which stem from local operator insertions were geometrically resummed; see~\cite{BBK1}.

\smallskip

\noindent\textbf{Remark.} \textbf{(1)} After splitting the integral~\eqref{eq:A9} (in particular, the $k$ summands), the integrands fit into the input class of the multivariate Almkvist-Zeilberger algorithm. Hence, if the split integrals are properly defined, they
obey homogeneous recurrence relations
in $N$ due to the existence theorems in~\cite{ZEILB}. However, so far we
failed to compute these recurrences due to time and space limitations.\\
\noindent\textbf{Remark.} \textbf{(2)} Usually the calculation of ${\cal I}(\ep,N)$ for
fixed integer  values of $N$ is a simpler task. If sufficiently many of these values are
known, one may guess these recurrences and with this input derive closed forms for
${\cal I}(\ep,N)$ using the
techniques applied in \cite{BKKS}. This has been illustrated for a large class of
3-loop quantities. However, at present no method is known to calculate the
amount of moments needed.

\smallskip

\noindent The $y_i$-integrals finally turn into  Euler integrals.
Here we outline a general framework, although in practice, different algorithms are used in specific cases, cf.\ e.g.\ \cite{HYP1,HYP2}.
To compute the integrals (\ref{eq:A9}) over the variables
$y_i$ we proceed as follows:
\begin{itemize}
\item decompose the denominator function using Mellin--Barnes integrals, see \cite{MB} and references therein,
\item decompose the
numerator functions, if needed, into multinomial series.
\end{itemize}
The denominator function has the structure
\begin{equation*}
[Q(y)]^{\beta(\varepsilon)} = \left[\sum_{k=1}^n
q_k(y)\right]^{\beta(\varepsilon)}~,
\end{equation*}
with $q_k(y)=a_1\dots a_m$ where $a_i\in\set\{1,y_i,1-y_i\}$ for $1\leq i\leq m$.
This function can be decomposed applying its Mellin-Barnes integral
representation $(n-1)$ times,

\begin{equation}
\frac{1}{(A+B)^q} = \frac{1}{2 \pi i}
\int_{\gamma-i\infty}^{\gamma+i\infty}~d\sigma~A^\sigma~B^{-q-\sigma}
\frac{\Gamma(-\sigma) \Gamma(q+\sigma)}{\Gamma(q)}~.
\label{MBT}
\end{equation}
Here $\gamma$ denotes the real part of the contour. Often Eq.~(\ref{MBT}) has to be
considered in the sense of its analytic continuation, see \cite{WIWA}.
The numerator factors $[P_{i,l}(y)]^{\alpha_{i,l}(\varepsilon,N)}$ obey

\vspace*{-0.5cm}

\begin{equation*}
[P_{i,l}(y)]^{\alpha_{i,l}(\varepsilon,N)} = \left[\sum_{k=1}^w
p_k(y)\right]^{\alpha_{i,l}(\varepsilon,N)}~,
\end{equation*}
where the monomials $p_k(y)$ have the same properties as $q_k(y)$. One expands
\begin{equation*}
[P_{i,l}(y)]^{\alpha_{i,l}(\varepsilon,N)} = \sum_{\substack{k_1,\dots,k_{w-1}\geq0}}
{\alpha_{i,l}(\varepsilon,N) \choose k_1, \ldots, k_{w-1}} \prod_{l=1}^{w-1} p_l(y)^{k_l}
p_w(y)^{\alpha_{i,l}(\varepsilon,N)-\sum_{r=1}^{w-1}k_r}.
\end{equation*}
Now all integrals over the variables $y_j$ can be computed by using the formula
\begin{equation*}
\int_0^1 dy y^{\alpha - 1} (1-y)^{\beta - 1} = B(\alpha,\beta) =
\frac{\Gamma(\alpha) \Gamma(\beta)}{\Gamma(\alpha + \beta)}
\end{equation*}
and one obtains
\begin{equation}\label{eq:A3}
\begin{split}
{\cal I}(\ep,N) =
\frac{1}{(2\pi i)^n}
&\int_{\gamma_1-i\infty}^{\gamma_1+i\infty} d \sigma_1 \ldots 	
\int_{\gamma_n-i\infty}^{\gamma_n+i\infty} d \sigma_n\\
&\sum_{k_1=1}^{L_1(N)} ... \sum_{k_v=1}^{L_v(N,k_1, ..., k_{v-1})}
\sum_{k=1}^l C_k(\ep, N, M)
\frac{\Gamma(z_{1,k}) \ldots \Gamma(z_{u,k})}	
     {\Gamma(z_{u+1,k}) \ldots \Gamma(z_{v,k})};
\end{split}	
\end{equation}
\normalsize
$l\in\set N$ and the summation over $k_i$ comes from the multinomial sums,
i.e., the upper bounds $L_1(N),\dots,L_{v}(N,k_1,\dots,k_{v-1})$ are integer linear in
the dependent parameters or $\infty$. Moreover, the
$z_{u,k}$ are linear functions with rational coefficients in terms of $\ep$, the Mellin-Barnes
integration variables $\sigma_1,\dots,\sigma_n$, and
the summation variables $k_1,\dots,k_v$.

\vspace*{2mm}\noindent
{{\sf Step~2: Representation in multi--sums.}}
The Mellin-Barnes integrals are carried out applying the residue theorem in
Eq.~(\ref{eq:A3}). The following representation is obtained:
\begin{equation}
{\cal I}(\ep,N) = \sum_{n_1=1}^\infty ... \sum_{n_r=1}^\infty
\sum_{k_1=1}^{L_1(N)} ... \sum_{k_v=1}^{L_v(N,k_1, ..., k_{v-1})}
\sum_{k=1}^l C_k(\ep, N, M)
\frac{\Gamma(t_{1,k}) \ldots \Gamma(t_{v',k})}	
     {\Gamma(t_{v'+1,k}) \ldots \Gamma(t_{w',k})}.
\label{eq:A4}
\end{equation}
Here the $t_{l,k}$ are linear functions with rational coefficients in terms of the $n_1,\dots,n_r$, of the $k_1,\dots,k_v$, and of $\ep$. Note that the residue theorem
may imply more than one infinite sum per Mellin-Barnes integral, i.e.,
$r\geq n$.\\
In general, this approach leads to a highly nested multi-sum. Fixing the loop order of the Feynman integrals and restricting to certain special situations usually enables one to find sum representations with fewer summation signs. E.g., as
worked out in~\cite{sums-comp-description}, one can identify
the underlying sums in terms of generalized hypergeometric functions, i.e.,
the number of infinite sums are reduced to one or in some cases to zero.

\vspace*{2mm}\noindent
{{\sf Step~3: Laurent series in $\varepsilon$.}}
Eq.~(\ref{eq:A4}) can now be expanded in the parameter $\ep$ using

\vspace*{-0.2cm}

\begin{equation}
\Gamma(n+1+\bar{\ep}) = \frac{\Gamma(n) \Gamma(1+\bar{\ep})}{B(n,1+\bar{\ep})}
\label{eq:A6}
\end{equation}
with $\bar{\ep} = r \ep$ for some $r\in\set Q$ and

\vspace*{-0.3cm}

\begin{equation}\label{Equ:Expan}
B(n, 1 + \bar{\ep}) = \frac{1}{n}\exp\left(\sum_{k=1}^\infty \frac{(-\bar{\ep})^k}{k} S_k(n)\right)
= \frac{1}{n}\sum_{k=0}^\infty (-\bar{\ep})^k S_{\underbrace{\mbox{\scriptsize 1, \ldots
,1}}_{\mbox{\scriptsize
$k$}}}(n)
\end{equation}
and other well-known transformations for the $\Gamma$-functions. Here the harmonic sums
$S_{\vec{a}}(N)$ \cite{BK,Vermaseren:99} for $N\in\set N$ are recursively defined by
\begin{equation}\label{Equ:HSums}
S_{b,\vec{a}}(N) = \sum_{k=1}^N \frac{({\rm sign}(b))^k}{k^{|b|}} S_{\vec{a}}(k),~~~
S_{\emptyset} = 1~.
\end{equation}
Note that in (\ref{eq:A6}) $n$ may stand for a
linear combination of parameters with coefficients in ${\mathbb Q}$. In case
of non-integer weight factors $r_i$ for the parameters in $n$ analytic continuations of harmonic sums
have to be considered
\cite{ANCONT1,ANCONT2,ANCONT3,ANCONT4}.
In case that $n$ is not an integer one may shift
to $n \rightarrow k~n \in \mathbb{N}$, which leads to the usual definition of the
harmonic sums in (\ref{Equ:Expan}). However, the summation operators have now to be
generalized and one usually ends up with {\it cyclotomic} harmonic sums worked out in~\cite{ABS2012}.

\noindent Applying~\eqref{eq:A6} with~\eqref{Equ:Expan} to each factor in~\eqref{eq:A4} produces for some $L>0$ the expansion
\begin{equation}\label{Equ:LaurentExp}
{\cal I}(\ep,N) = \sum_{l = -L}^\infty \ep^l I_l(N);
\end{equation}
$L$ equals the loop order in case of infra-red finite integrals; otherwise, $L$
may be larger.

\vspace*{-0.1cm}

\begin{remark}\label{Remark:ExpansionExist}
In order to guarantee correctness of this construction, i.e., performing the expansion first on the summand level of~(\ref{eq:A4}) and afterwards applying the summation on the coefficients of the summand expansion (i.e., exchanging the differential operator $D_{\ep}$ and the summation quantifiers) analytic arguments have to be considered. For all our computations this construction was possible.
\end{remark}

\vspace*{-0.1cm}

\noindent The general expression of the functions $I_l(N)$ in terms of nested sums are
\begin{equation}\label{eq:A5}
\begin{split}
I_l(N) =\sum_{n_1=1}^\infty ... \sum_{n_r=1}^\infty
&\sum_{k_1=1}^{L_1(N)} ... \sum_{k_v=1}^{L_v(N,k_1, ..., k_{v-1})} \sum_{j=1}^s H_j(N;n_1, ..., n_r; ,k_1, ..., k_{v})\\
& \times\prod_i S_{\vec{a}_{i,j}}(L_{i,j}(N;n_1, ...,n_r; ,k_1,..., k_{v}));
\end{split}
\end{equation}

\vspace*{-0.3cm}

\noindent $H_j(N;n_1, ... , k_{v})$ denote proper hypergeometric terms\footnote{For a precise definition of proper hypergeometric terms we refer, e.g., to~\cite{Wegschaider}. For all our applications it suffices to know that $H_j$ might be a product of Gamma-functions (occurring in the numerator and denominator) with linear dependence on the variables $N,n_i,k_i$ times a rational function in these variables where the denominator factors linearly.}\label{Footnote:Proper} and
$S_{\vec{a}_{i,j}}(L_{i,j}(N;n_1, ..., k_{v}))$ are harmonic sums with the
index set $\vec{a}_{i,j}$ and $L_{i,j}$ (usually integer linear) functions of the arguments $(N;n_1, ..., k_{v})$. The
sum-structure
in (\ref{eq:A5}) is usually obtained performing the synchronization of arguments, see
\cite{Vermaseren:99},
and applying the associated quasi--shuffle algebra, see \cite{ALGEBRA}.

\section{First approach to the problem}\label{Equ:FirstApproach}

In the following we limit the investigation to a sub-class of integrals of the type
(\ref{eq:A7}) and consider two- and simpler three-loop diagrams, which occurred
in the calculation of the massive Wilson coefficients for deep-inelastic scattering; see~\cite{HYP2,HWILS1,HWILS2,HWILS3,sums-comp-description}.
Looking at the reduction steps of the previous section we obtain the following result. If we succeed in finding the representation~\eqref{Equ:LaurentExp} with~\eqref{eq:A5} it follows constructively that for each $N\in\set N$ with $N\geq\lambda$ for some $\lambda\in\set N$ the integral ${\cal I}(\ep,N)$ has a Laurent expansion in $\ep$ and thus it is an analytic function in $\varepsilon$ throughout an annular region centered at $0$ where the pole at $\ep=0$ has order $L$.
In~\cite{sums-comp-description,HYP2,SchneiderSummation:10} we started with the sum representation of the coefficients~\eqref{eq:A5} and the main task was to simplify the expressions in terms of harmonic sums.
In this article, we follow a new approach that directly attacks the sum representation~\eqref{eq:A4} and searches for the first coefficients of its $\ep$-expansion~\eqref{Equ:LaurentExp}.
By splitting~\eqref{eq:A4} accordingly (and pulling
$C_k(\ep,N,m)$) our integral can be written as a linear combination of
hypergeometric multi-sums of the following form.

\begin{assumption}\label{Assum:SumProp}

{\color{white}.}

\vspace*{-0.65cm}

\begin{equation}\label{physics-sums-general}
\mathcal{S}(\ep, N) = \sum\limits_{\sigma_1=p_1}^\infty \cdots
\sum\limits_{\sigma_s=p_s}^\infty
\sum\limits_{j_0=q_0}^{L_1(N)} \sum\limits_{j_1=q_1}^{L_2(N,j_0)}
\cdots \sum\limits_{j_r=q_r}^{L_r(N,j_0,\dots,j_{r-1})} \mathcal{F} \left( N,  \sigma,j_0,\dots,j_{r-1}),
\ep\right)
\end{equation}

\vspace*{-0.3cm}

\noindent where
\begin{enumerate}
\item $N\in\set N$ with $N \geq \lambda$ for some given $\lambda\in\set N $,
$\ep >0$ is a real parameter;
\item  the upper summation bounds $L_l(N,j_0,\dots,j_{l-1})$ are integer linear in $N,j_0,\dots,j_{l-1}$, and the lower bounds are given constants $p_i, q_l \in \mathbb{N}$ for all $1\leq i \leq s$ and $0 \leq l \leq r$;
\item $\mathcal{F}$ is a proper hypergeometric term (see Footnote~\ref{Footnote:Proper}) with respect to the
integer variable $N$ and all summation variables $(\sigma, j)=(\sigma_1,\dots,\sigma_s,j_0,\dots,j_r)\in
\mathbb{Z}^{s+r+1}$.
\end{enumerate}
\end{assumption}

\begin{remark}\label{Remark:SplittingSums}
While splitting the sum~\eqref{eq:A4} into sums of the form~\eqref{physics-sums-general} it might happen
that the infinite sums over individual monomials diverge for fixed values of $\varepsilon$,
despite the convergence of the complete
expression.
We will deal with these cases in Section~\ref{Sec:EfficientApproach} and consider only sums
which are convergent at the moment.
\end{remark}

In other words, we assume that~\eqref{physics-sums-general} itself is analytic in $\varepsilon$ throughout an annular region centered at $0$ and we try to find the first coefficients $F_{t}(N),F_{t+1}(N),\dots,F_u(N)$ in terms of
indefinite nested product-sum expressions of its expansion
\begin{equation}\label{Equ:FExp2}
\mathcal{S}(\ep, N) = F_t(N)\varepsilon^{t}+F_{t+1}(N)\varepsilon^{t+1}+F_{t+2}(N)\varepsilon^{t+2}+\dots.
\end{equation}
with $t\in\set Z$. In all our computations it turns out that the summand $\mathcal{F}\left( N, \sigma, j, \ep\right)$ satisfies besides properties (1)--(3) the following asymptotic behavior:

\smallskip

\begin{enumerate}
\item[(4)] for all $1 \leq i \leq s$ we have
\begin{equation}\label{limitconsiderations}
 \mathcal{F} \left( N, \sigma, j, \ep\right) =
\mathcal{O}\left( \sigma_i^{-d_i} e^{-c_i
\sigma_i}\right)\quad \text{as} \quad
\sigma_i\rightarrow \infty \quad \text{with} \quad c_i \geq 0, \quad
d_i > 0;
\end{equation}
\end{enumerate}
for simplicity we do not consider the log-parts. For later considerations in Section~\ref{Sec:FindRecurrences} we suppose that such constants $c_i$ and $d_i$ are given explicitly. E.g.,
using the behavior \cite[Section 13.6]{WIWA} of  $\log \Gamma(z)$
for large $\left| z \right|$
in the region where $\left| \arg(z) \right| < \pi$ and  $\left|
\arg(z+a) \right|< \pi$:
\begin{equation}\label{asymptGamma}
 \log \Gamma(z+a)= (z+a-\frac{1}{2})\log z-z+ \mathcal{O}(1),
\end{equation}
such constants can be easily computed. If not all $c_i>0$ for $1\leq i\leq s$, things get more complicated and --for simplicity-- we restrict ourselves to the case that $s=1$ and $c_1=0$; we refer again to Section~\ref{Sec:EfficientApproach} for further details how one can treat the more general case.

\smallskip

\begin{enumerate}
\item[(5)]
If $s=1$ and $c_1=0$, we suppose that we are given a constant $o\in\set N$ such that
\begin{equation}\label{Equ:SumDefined}
\mathcal{S}(\ep, N) = \sum\limits_{\sigma_1=p_1}^{\infty}\sigma_1^o\mathcal{F} \left( N,  \sigma_1,j,\ep\right)
\end{equation}
converges absolutely for any small nonzero $\ep$ around $0$, $N\geq B$ and any $j$ that runs over the summation range.
\end{enumerate}
Using, e.g., facts about hypergeometric functions from~\cite[Thm.~2.1.1]{AAR} the maximal such constant $o$ in~\eqref{Equ:SumDefined} can be determined.

\begin{example}
The following sum is a typical entry from the list of sum
representations for a class of Feynman parameter integrals we computed:
\small
\begin{equation} \label{Ex-857}
\begin{split}
 &\mathcal{U} \left(\ep,N\right):=
\sum\limits_{ \sigma_1 =0}^{\infty}
\sum\limits_{ j_0=0}^{N-3}
\sum\limits_{ j_1=0}^{N-j_0-3}
\sum\limits_{ j_2=0}^{j_0+1}
\binom{j_0 +1}{ j_2} \frac{\binom{ N  - j_0 -3}{j_1}\Gamma (j_1+j_2+2) \Gamma (j_1+j_2+3)}{N!}\\
& \times
\frac{(-1)^N\left(\frac{\ep}{2}+1\right)_{\sigma_1 } (-\ep)_{\sigma_1 }
    (j_1+ j_2+3)_{\sigma_1}  \left(3-\frac{\ep
}{2}\right)_{j_1}}{(j_1+4)_{\sigma_1 }
    \left(-\frac{\ep}{2}+j_1+j_2+4\right)_{\sigma_1}
\left(4-\frac{\ep}{2}\right)_{j_1+j_2}}
\frac{
    \Gamma (N-j_0-1) \Gamma (N-j_1-j_2-1)
   }{\Gamma(\sigma_1 +1) \Gamma(j_1+4)
    \Gamma (N-j_0-2) };
\end{split}
\end{equation}
\normalsize
we denote by $(x)_k=x(x+1)\dots(x+k-1)$ the Pochhammer symbol defined for non-negative integers $k$. Then using formulas such as $(x)_k=\Gamma(x+k)/\Gamma(x)$ and $\binom{x}{k}=\Gamma(x+1)/\Gamma(x-k+1)/\Gamma(k+1)$ and applying~\eqref{asymptGamma} we get the asymptotic behavior $\mathcal{O}(\sigma_1^{-5})$ of the summand. Moreover, we choose the maximal $o=3$ such that condition~\eqref{Equ:SumDefined} is satisfied.
\end{example}

\noindent Subsequently, we will develop an algorithm that finds, whenever possible, representations for the coefficients in the expansion~\eqref{Equ:FExp2}  in terms of indefinite nested sums and products\footnote{This means in particular indefinite nested sums over hypergeometric terms (like binomials, factorials, Pochhammer symbols) that may occur as polynomial expressions with the additional constraint that the summation index $i_j$ of a sum $\sum_{i_j=1}^{i_{j+1}}f(i_j)$ may occur only as the upper index of its inner sums and products, but not inside the inner sums themselves; for a formal but lengthy definition see~\cite{Schneider:11}. Typical examples are sums of the form~\eqref{Equ:HSums} above, or of the forms~\eqref{Equ:HomSol} and~\eqref{Equ:InHomSol} given below.}.

\begin{theorem}\label{Thm:AlgForMultiSum}
Let ${\cal S}(\varepsilon,N)$ be a sum with properties (1)--(5) from Assumption~\ref{Assum:SumProp}
which forms an analytic function in $\varepsilon$ throughout an annular region centered at $0$  with the Laurent expansion~\eqref{Equ:FExp2} for some $t\in\set Z$
for each nonnegative $N$; let $u\in\set N$. Then there is an algorithm which finds the maximal $r\in\{t-1,t,\dots,u\}$ such that the $f_t(N),\dots,f_r(N)$ are expressible in terms of indefinite nested product-sums; it outputs such expressions $F_t(N),\dots,F_r(N)$ and $\lambda\in\set N$ s.t.\ $f_i(k)=F_i(k)$ for all $0\leq i\leq r$ and all $k\in\set N$ with $k\geq\lambda$.
\end{theorem}

\noindent This result is based on the fact that such sums ${\cal S}(\varepsilon,N)$ satisfy a recurrence relation.

\begin{example}\label{Exp:SingleExt1}
Consider the single nested sum
\begin{equation}\label{Equ:SimpleSumProblem}
\mathcal{S}(\ep, N)=\sum_{k=0}^{N-1} \frac{(-2)^{k} (k+2) \Gamma (4-\varepsilon) \Gamma \big(\frac{\varepsilon}{2}+3\big) \Gamma (N) \Gamma\big(-\frac{\varepsilon}{2}+k+2\big)}{\Gamma \big(2-\frac{\varepsilon}{2}\big) \Gamma (-\varepsilon+k+4) \Gamma\big(\frac{\varepsilon}{2}+k+3\big) \Gamma(N-k)}
\end{equation}
over a proper hypergeometric term; note that an expansion~\eqref{Equ:FExp2} with $t=0$ exists following the arguments from Remark~\ref{Remark:ExpansionExist}. In the first step we compute the recurrence relation
\begin{equation}\label{Equ:SingleSumRec}
a_0(\varepsilon,N)\mathcal{S}(\ep, N)+a_1(\varepsilon,N) \mathcal{S}(\ep, N+1)+a_2(\varepsilon,N)\mathcal{S}(\ep, N+2)=h(\ep,N)
\end{equation}
with
\small
\begin{equation}\label{Equ:SingleSumRecCoef}
\begin{split}
h(\ep,N)&=-24 N-48+(2 N-20)\varepsilon+(2 N+6)\varepsilon^2+2\varepsilon^3,\\
a_0(\varepsilon,N)&=2 N (N+1) (\varepsilon+2 N+5),\quad
a_1(\varepsilon,N)=(N+1) \big(\varepsilon^2+2 \varepsilon N+5 \varepsilon+4 N+12\big),\\
a_2(\varepsilon,N)&=(\varepsilon-N-4) (\varepsilon+2 N+3) (\varepsilon+2 N+6)
\end{split}
\end{equation}
\normalsize
which holds for all $N\geq1$.
This task can be accomplished for instance by the packages~\cite{Paule-Schorn}, \cite{Wegschaider} or~\cite{Sigma-Manual} which are based on the creative telescoping paradigm presented in~\cite{Zeilberger-identities} or the paradigm presented in~\cite{Celine}. Then together with the first two initial values for $N=1,2$,
\begin{equation}\label{Equ:SingleSumInitial}
\mathcal{S}(\ep, 1)=2\quad\quad\text{and}\quad\quad\mathcal{S}(\ep,2)=2-\frac{6}{\varepsilon+6}=1+\frac{1}{6}\varepsilon-\frac{1}{36}\varepsilon^2+O(\varepsilon^3),
\end{equation}
we will be able to compute, e.g., the sum representations of the first 2 coefficients
\begin{align}\label{Equ:SingleSumF0}
F_0(N)&=\tfrac{3 \big(2 N^2+4 N+1\big)}{2 N (N+1) (N+2)}-\tfrac{3 (-1)^N}{2 N (N+1)(N+2)},\\
\label{Equ:SingleSumF1}
F_1(N)&=\tfrac{10 N^3+52 N^2+63 N+10}{8 N (N+1) (N+2)^2}-\tfrac{3 S_1(N)}{2 N
(N+2)}+\tfrac{3 S_{-1}(N)}{2 N (N+2)}+\tfrac{(-1)^N (N-10)}{8 N (N+1) (N+2)^2};
\end{align}
of the $\ep$-expansion~\eqref{Equ:FExp2}
with $t=0$; for more details see Examples~\ref{Exp:SingleExt2} and~\ref{Exp:SingleExt3}.
\end{example}

In Subsection~\ref{Sec:RecurrenceSolver} we will develop a recurrence solver which finds the representation of the $F_i(N)$ from~\eqref{Equ:FExp2} in terms of indefinite nested sums and products whenever this is possible. Afterwards, we combine all these methods to prove Theorem~\ref{Thm:AlgForMultiSum} in Subsection~\ref{Sec:EffectiveMethod}.

\subsection{A recurrence solver for $\varepsilon$-expansions}\label{Sec:RecurrenceSolver}

Restricting the $\mathcal{O}$-notation to formal Laurent series $f=\sum_{i=r}^{\infty}f_i\ep^i$ and $g=\sum_{i=s}^{\infty}g_i\ep^i$, the notation
$f=g+O(\ep^t)$
for some $t\in\set Z$ means that the order of $f-g$ is larger or equal to $t$, i.e., $f-g=\sum_{i=t}^{\infty}h_i\ep^i$. Subsequently, $\set K$ denotes a field with $\set Q\subseteq\set K$ in which the usual operations can be computed. We start with the following

\begin{lemma}\label{Lemma:UniqueSeqSol}
Let $\mu\in\set N$, and let $a_0(\varepsilon,N),\dots,a_d(\varepsilon,N)\in\set K[\varepsilon,N]$ be such that $a_d(0,k)\neq0$ for all $k\in\set N$ with $k\geq\mu$. Let  $h_t,\dots,h_u:\set N\to\set K$ ($t,u\in\set Z$ with $t\leq u$) be functions, and let $c_{i,k}\in\set K$ with $t\leq i\leq u$ and $\mu\leq k<\mu+d$. Then there are unique functions $F_t,\dots,F_u:\set N\to\set K$ (up to the first $\mu$ evaluation points) such that $F_i(k)=c_{i,k}$ for all $t\leq i\leq u$ and $\mu\leq k< d+\mu$ and such that for $T(\varepsilon,N)=\sum_{i=t}^uF_i(N)\varepsilon^i$ we have
\begin{equation}\label{Equ:ExpansionEquMod}
a_0(\varepsilon,N)T(\varepsilon,N)+\dots+a_d(\varepsilon,N)T(\varepsilon,
N+d)=h_0(N)+\dots+h_u(N)\varepsilon^u+O(\varepsilon^{u+1})
\end{equation}
for all $N\geq\mu$. If the $h_i(N)$ are computable, the values of the $F_i(N)$ with $N\geq\mu$ can be computed by recurrence relations.
\end{lemma}
\begin{proof}
Plugging the ansatz $T(\varepsilon,N)=\sum_{i=t}^u F_i(N)\varepsilon^i$ into~\eqref{Equ:ExpansionEquMod} and doing coefficient comparison w.r.t.\ $\varepsilon^t$ yields the constraint
\begin{equation}\label{Equ:LinRecConst}
a_0(0,N)F_t(N)+\dots+a_d(0,N)F_t(N+d)=h_t(N).
\end{equation}
Since $a_d(0,N)$ is non-zero for any integer evaluation $N\geq\mu$, the function $F_0:\set N\to\set K$ is uniquely determined by the initial values $F_t(\mu)=c_{t,\mu},\dots,F_t(\mu+d-1)=c_{t,\mu+d-1}$ -- up to the first $\mu$ evaluation points; in particular the values $F_t(k)$ for $k\geq\mu$ can be computed by the recurrence relation~\eqref{Equ:LinRecConst}.
Moving the $F_t(N)\varepsilon^t$ in~\eqref{Equ:ExpansionEquMod} to the right hand side gives

\vspace*{-0.5cm}

\begin{multline*}
a_0(\varepsilon,N)\sum_{i=t+1}^uF_i(N)\varepsilon^i+\dots+a_d(\varepsilon,N)
\sum_{i=t+1}^uF_i(N+d)\varepsilon^i\\[-0.3cm]
=-\Big[a_0(\varepsilon,N)h_t(N)\ep^t+\dots+a_d(\varepsilon,N)h_t(N+d)\ep^t\Big]
+\sum_{i=t}^uh_i(N)\varepsilon^i;
\end{multline*}

\vspace*{-0.3cm}

\noindent denote the coefficient of $\varepsilon^i$ on the right side by $\tilde{h}_i$. Since the coefficient of $\ep^t$ on the left side is $0$, it is also $0$ on the right side and we can write
\begin{equation*}
a_0(\varepsilon,N)\sum_{i=t+1}^{u}F_{i}(N)\varepsilon^i+\dots+a_d(\varepsilon,N)
\sum_{i=t+1}^{u}F_{i}(N+d)\varepsilon^i=\sum_{i=t+1}^{u}\tilde{h}_{i}(N)\varepsilon^i + O(\varepsilon^{u+1})
\end{equation*}
for all $N\in\set N$ with $N\geq\mu$.
Repeating this process proves the lemma.
\end{proof}

\begin{example}\label{Exp:SingleExt2}
Consider the recurrence~\eqref{Equ:SingleSumRec} with the coefficients~\eqref{Equ:SingleSumRecCoef}. Then by Lemma~\ref{Lemma:UniqueSeqSol} there are unique functions $F_0(N)$ and $F_1(N)$ with $T(N)=F_0(N)+F_e(N)\ep$ such that $T(\ep,1)=2$, $T(\ep,2)=1+\tfrac{1}{6}\ep$ and
\begin{equation}\label{Equ:SingleSumRecMod}
a_0(\varepsilon,N)T(\ep, N)+a_1(\varepsilon,N) T(\ep, N+1)+a_2(\varepsilon,N)T(\ep, N+2)=h(\ep,N)+O(\ep^2)
\end{equation}
hold for $N\geq1$. In particular, by setting $\varepsilon=0$, we get
\begin{equation}\label{Equ:SingleSumConstrained}
a_0(0,N)F_0(N)+a_1(0,N)F_0(N+1)+a_2(0,N)F_0(N+2)=-24N-48;
\end{equation}
the values of $F_0(N)$ can be computed with~\eqref{Equ:SingleSumConstrained} and the initial values $F_0(1)=2,F_0(2)=1$.
\end{example}

\noindent At this point we exploit algorithms from~\cite{Hyper,Abramov:94,Schneider:01,Schneider-solving} which can constructively decide if a solution with certain initial values is expressible in terms of indefinite nested products and sums. To be more precise, with the algorithms implemented in~\SigmaP\ one can solve the following problem.

\medskip
\small
\begin{ProblemSpec}{\textbf{Problem \ProblemRS}: \textbf{R}ecurrence \textbf{S}olver for indefinite nested product-sum expressions.}
\textbf{Given} $a_0(N),\dots,a_d(N)\in\set K[N]$; given $\mu\in\set N$ such that $a_d(k)\neq0$ for all $k\in\set N$ with $N\geq\mu$; given an expression $h(N)$ in terms of indefinite nested product-sum expressions which can be evaluated for all $N\in\set N$ with $N\geq\mu$; given the initial values $(c_{\mu},\dots,c_{\mu+d-1})$ which produce the sequence $(c_i)_{i\geq\mu}\in{\set K}^{\set N}$ by the defining recurrence relation
$$a_0(N)c_N+a_1(N)c_{N+1}+\dots+a_d(N)c_{N+d}=h(N)\quad\forall N\geq\mu.$$
\textbf{Find}, if possible, $\lambda\in\set N$ with $\lambda\geq\mu$ and an indefinite nested product-sum expression $g(N)$ such that $g(k)=c_k$ for all $k\geq\lambda$.
\end{ProblemSpec}
\normalsize
\smallskip

\noindent\textit{Remark.} Later, we will give further details only for a special case that occurred in almost all instances of our computations related to Feynman integrals; see Theorem~\ref{Thm:FactorRec}.

\begin{example}\label{Exp:SingleExt3}
With the input $F_0(1)=2, F_0(2)=1$ and~\eqref{Equ:SingleSumConstrained} \SigmaP\ computes the solution~\eqref{Equ:SingleSumF0}.
Plugging this partial solution $T(\ep,N)=F_0(N)+\dots$ into~\eqref{Equ:SingleSumRecMod} and doing coefficient comparison leads to
\small
\begin{equation*}
\sum_{i=0}^2a_i(0,N)F_1(N+i)=\frac{-10 N^4-98 N^3-344 N^2-511 N-267}{(N+2) (N+3)(N+4)}-\frac{3 (-1)^N (3 N+7)}{(N+2) (N+3)(N+4)}.
\end{equation*}
\normalsize
Then together with $F_1(1)=0,F_1(2)=1/6$, \SigmaP\ finds~\eqref{Equ:SingleSumF1}. Since also~\eqref{Exp:SingleExt1} satisfies~\eqref{Equ:SingleSumRecMod} with the same initial values~\eqref{Equ:SingleSumInitial}, the first two coefficients of the expansion of~\eqref{Exp:SingleExt1} are equal to $F_0(N)$ and $F_1(N)$ by Lemma~\ref{Lemma:UniqueSeqSol}.
\end{example}

\noindent This iterative procedure can be summarized as follows.
\smallskip

\small

\noindent\textbf{Algorithm~\FLSR} (\textbf{F}ormal \textbf{L}aurent \textbf{S}eries solutions of linear \textbf{R}ecurrences)\\
\textbf{Input:} $\mu\in\set N$; $a_0(\varepsilon,N),\dots,a_d(\varepsilon,N)\in\set K[\varepsilon,N]$ such that $a_d(0,k)\neq0$ for all $k\in\set N$ with $k\geq\mu$; indefinite nested product-sum expressions $h_t(N)$, $\dots$, $h_u(N)$ ($t,u\in\set Z$ with $t\leq u$) which can be evaluated for all $N\in\set N$ with $N\geq \mu$;
$c_{i,j}\in\set K$ with $t\leq i\leq u$ and $\mu\leq j<\mu+d$.\\
\textbf{Output} $(r,\lambda,\tilde{T}(N))$:
The maximal number $r\in\{t-1,0,\dots,u\}$ s.t.\ for the unique solution  $T(N)=\sum_{i=t}^uF_i(N)\varepsilon^i$ with $F_i(k)=c_{i,k}$ for all $\mu\leq k<\mu+d$ and with the relation~\eqref{Equ:ExpansionEquMod} the following holds: there are indefinite nested product-sum expressions that are equal to $F_t(N),\dots,F_r(N)$ for all $N\geq\lambda$ for some $\lambda\geq\mu$; if $r\geq0$, return such an expression $\tilde{T}(N)$ for $T(N)$ together with $\lambda$.

\vspace*{-0.3cm}

\begin{enumerate}
\item (Preprocessing) By Lemma~\ref{Lemma:UniqueSeqSol} we can compute as many initial values $c_{i,k}:=F_i(k)$ for $k\geq\mu$ as needed for the steps given below (at most $\lambda-\mu$ extra values are needed).

\vspace*{-0.3cm}

\item Set $r:=t$, $\lambda:=\mu$, and $\tilde{T}(N):=0$.

\vspace*{-0.3cm}

\item Note that $(F_r(N))_{N\geq\mu}$ is defined by the initial values $F_r(N)$ ($\lambda\leq N< d+\lambda$) and the recurrence
\begin{equation}\label{KeyEqu}
a_0(0,N)F_r(N)+\dots+a_d(0,N)F_r(N+d)=h_r(N)
\end{equation}
for all $N\in\set N$ with $N\geq\lambda$; see the proofs of Lemma~\ref{Lemma:UniqueSeqSol} or Theorem~\ref{Thm:ExpansionAlg}. By solving problem~\ProblemRS\
decide constructively if there is a $\lambda'\geq\lambda$ such that $F_r(N)$ can be computed in terms of an indefinite nested product-sum expression  $\tilde{F}_r(N)$ for all $N\in\set N$ with $N\geq \lambda'$.

\vspace*{-0.3cm}

\item If this fails, RETURN $(r-1,\lambda,\tilde{T}(N))$. Otherwise, set $\tilde{T}(N):=\tilde{T}(N)+\tilde{F}_r(N)\varepsilon^r$.

\vspace*{-0.3cm}

\item If $r=u$, RETURN $(r,\lambda, \tilde{T}(N))$.

\vspace*{-0.3cm}

\item Collect the coefficients (product-sum expressions) w.r.t.\ $\ep^i$ for all $i$ ($r+1\leq i\leq u$):

\vspace*{-0.6cm}

\begin{equation*}
h'_i(N):=\text{coeff}(-\Big[a_0(\varepsilon,N)F_r(N)+\dots+a_d(\varepsilon,N)F_r(N+d)\Big]
+\sum_{i=r+1}^uh_i(N)\varepsilon^{i},\varepsilon^i).
\end{equation*}

\vspace*{-0.5cm}

\item Set $h_i:=h'_i$ for all $r+1\leq i\leq u$, set $r:=r+1$ and GOTO Step~3.
\end{enumerate}
\normalsize

\begin{theorem}\label{Thm:ExpansionAlg}
The algorithm terminates and fulfills the input--output specification.
\end{theorem}
\begin{proof}
We show that entering the $r$th iteration of the loop ($r\geq t$) we have for all $N\geq\lambda$ that
\begin{equation}\label{Equ:RecExpansion}
a_0(\varepsilon,N)\sum_{i=r}^uF_i(N)\varepsilon^i+\dots+
a_d(\varepsilon,N)\sum_{i=r}^uF_i(N+d)\varepsilon^i
=\sum_{i=r}^uh_i(N)\varepsilon^{i}+O(\varepsilon^{u+1})
\end{equation}
where the $h_r(N),\dots,h_u(N)$ are given explicitly in terms of indefinite nested product-sum expressions. Moreover, we show that the obtained expression $\tilde{T}(N)=\sum_{i=t}^{r-1}\tilde{F}_i(N)\varepsilon^i$ equals the values $\sum_{i=t}^{r-1}F_i(N)\varepsilon^i$ for each $N\geq\lambda$. For $r=t$ this holds by assumption. Now suppose that these properties hold when entering the $r$th iteration of the loop ($r\geq t$).
Then coefficient comparison in~\eqref{Equ:RecExpansion} w.r.t.\ $\varepsilon^r$  yields the constraint~\eqref{KeyEqu}
for all $N\geq\lambda$ as claimed in Step~3 of the algorithm. Solving problem~\ProblemRS\ decides constructively if there is a $\lambda'\geq0$ such that $F_r(N)$ can be computed by an expression in terms of indefinite nested product-sum expressions, say $\tilde{F}_r(N)$, for all $N$ with $N\geq \lambda'$.
If this fails, $F_r(N)$ cannot be represented with such an expression and the output $(r-1,\lambda,\tilde{T}(N))$ with $\tilde{T}(N)=\sum_{i=t}^{r-1}\tilde{F}_i(N)$ is correct. Otherwise, the indefinite nested product-sum expressions $\tilde{F}_i(N)$ for $t\leq i\leq r$ give the values $F_i(N)$ for all $N\in\set N$ with $N\geq\lambda'$. Now move the term $F_r(N)\varepsilon^{r}$ in~\eqref{Equ:RecExpansion} to the
right hand side and replace it with $\tilde{F_r}(N)\varepsilon^{r}$. This gives

\vspace*{-0.5cm}

\begin{align*}
a_0(\varepsilon,N)\sum_{i=r+1}^u&F_i(N)\varepsilon^i+\dots+
a_d(\varepsilon,N)\sum_{i=r+1}^uF_i(N+d)\varepsilon^i
=-\sum_{i=0}^da_i(\varepsilon,N)\tilde{F}_r(N+i)\\[-0.25cm]
&+\sum_{i=r}^uh_i(N)\varepsilon^{i}+O(\ep^{u+1})=:\tilde{h}_{r+1}(N)\varepsilon^{r+1}+\dots+\tilde{h}_{u}(N)\varepsilon^{u}+O(\ep^{u+1})
\end{align*}

\vspace*{-0.15cm}

\noindent for all $N\geq\lambda'$ where
$\tilde{h}_{r+1}(N),\dots,\tilde{h}_u(N)$ are given in terms of indefinite nested product-sum expressions that can be evaluated for all $N\in\set N$ with $N\geq\lambda'$. By redefining the $h_i(N)$ as in Step~7 of the algorithm we obtain the relation~\eqref{Equ:RecExpansion} for the case $r+1$.
\end{proof}

\noindent Algorithm~\FLSR\ has been implemented within the summation package~\texttt{Sigma}. E.g., the expansion for the sum~\eqref{Equ:SimpleSumProblem} with $s=0$, $t=1$ and $\texttt{start}=1$ is computed by
\small
\begin{multline*}
\texttt{GenerateExpansion}[a_0(\varepsilon,N)S[N]+a_1(\varepsilon,N) S[N+1]+a_2(\varepsilon,N)S[N+2],\\
\{-24 N-48,2 N-20\},S[N], \{\ep, s, t\},
\{\texttt{start},\{\{2, 1\}, \{0, 1/6\}\}\}];
\end{multline*}
\normalsize
here the $a_i(\varepsilon,N)$ stand for the polynomials~\eqref{Equ:SingleSumRecCoef}, $\{-24 N-48,2 N-20\}$ is the list of the first coefficients on the right hand side of~\eqref{Equ:SingleSumRec}, and $\texttt{start}$ tells the procedure that the list of initial values $\{\{2, 1\}, \{0, 1/6\}\}$ from~\eqref{Equ:SingleSumInitial} corresponds to $N=1,2$.

As demonstrated already in Example~\ref{Exp:SingleExt3} the following application is immediate.

\begin{corollary}\label{Cor:ExpansionSolver}
For each nonnegative $N$, let ${\cal S}(\varepsilon,N)$ be an analytic function in $\varepsilon$ throughout an annular region centered at $0$  with the Laurent expansion ${\cal S}(\varepsilon,N)=\sum_{i=t}^{\infty}f_i(N)\varepsilon^i$ for some $t\in\set Z$, and suppose that ${\cal S}(\varepsilon,N)$ satisfies the recurrence~\eqref{Equ:ExpansionEquMod}
with coefficients and inhomogeneous part as stated in Algorithm~\FLSR\ for some $\mu\in\set N$; define $c_{i,k}:=F_i(k)$ for $t\leq i\leq u$ and $\mu\leq k<\mu+d$. Let $(r,\lambda, \sum_{i=t}^rF_i(N)\varepsilon^i)$ be the output of
Algorithm~\FLSR.
Then $f_i(k)=F_i(k)$ for all $t\leq i\leq r$ and all $k\in\set N$ with $k\geq\lambda$.
\end{corollary}

\noindent For further considerations we restrict to the following special case.
We observed --to our surprise-- in almost all examples arising from
Feynman integrals that the operator
\begin{equation}\label{FactorOp}
\sum_{i=0}^da_i(0,N)S_N^i=c(N)(S_N-b_d(N))(S_N-b_{d-1}(N))\dots(S_N-b_1(N))
\end{equation}
with the shift operator $S_N$ factorizes completely for some $b_1,\dots,b_d,c\in\set K(N)$; the rational functions can be computed by Petkov{\v s}ek's algorithm~\cite{Hyper}. In this particular instance we can construct immediately the complete solution space of
\begin{equation}\label{KeyEquGeneric}
a_0(0,N)F(N)+\dots+a_d(0,N)F(N+d)=X(N)
\end{equation}
for a generic sequence $X(N)$. Namely,
choose $\mu_i\in\set N$ such that the numerator and denominator polynomial of $b_i(j)$ have no zeros for all evaluations $j\in\set N$ with $j\geq\mu_i$, and take $\lambda:=\max_{1\leq i\leq d}\mu_i+1$. Now
define for $1\leq i\leq d$ the hypergeometric terms
$h_i(N)=\prod_{j=\lambda}^Nb_i(j-1)$. Then by~\cite{Abramov:94} one gets the $d$ linearly independent solutions

\vspace*{-0.6cm}

\begin{equation}\label{Equ:HomSol}
H_1(N):=h_1(N), \dots,H_d(N):=h_1(N)\sum_{i_1=\lambda}^{N-1}\!\!\frac{h_2(i_1)}{h_1(i_1+1)}\dots\sum_{
i_{d-1}=\lambda}^{i_{d-2}-1}\!\!\frac{h_d(i_{d-1})}{h_d(i_{d-1}+1)}
\end{equation}
of the homogeneous version of~\eqref{KeyEquGeneric}, and the particular solution
\begin{equation}\label{Equ:InHomSol}
P(N):=\frac{h_1(N)}{c(N)}\sum_{i_1=\lambda}^{N-1}\frac{h_2(i_1)}{h_1(i_1+1)}
\dots\sum_{i_{d-1}=\lambda}^{i_{d-2}-1}\frac{h_d(i_{d-1})}{h_{d-1}(i_{d-1}+1)}
\sum_{i_{d}=\lambda}^{i_{d-1}-1}\frac{X(i_{d})}{h_d(i_{d}+1)}
\end{equation}
of~\eqref{KeyEquGeneric} itself. In other words, the solution space of~\eqref{KeyEquGeneric} is explicitly given by
\begin{equation}\label{Equ:SolutionSpace}
\{c_1\,H_1(N)+\dots+c_d\,H_d(N)+P(N)|c_1\dots,c_d\in\set K\};
\end{equation}
here the nesting depth (counting the nested sums) of $H_i$ is $i-1$ and of $P$ is $d$.

Given this explicit solution space~\eqref{Equ:SolutionSpace} we end up with the following result.

\begin{theorem}\label{Thm:FactorRec}
Let $h_t(N),h_{t+1}(N),\dots$ with $t\in\set Z$ be functions that are computable in terms of indefinite nested product-sum expressions
where the nesting depth of the summation quantifiers of $h_i(N)$ is $d_i$; let $a_i(\varepsilon,N)\in\set K[\varepsilon,N]$ be such that the operator factors as in~\eqref{FactorOp} for some $c,b_i\in\set K(N), c\neq0$. If ${\cal S}(\varepsilon,N)=\sum_{i=t}^{\infty}F_i(N)\varepsilon^i$ is a solution
of
\begin{equation}\label{Equ:ExpansionEqu}
a_0(\varepsilon,N){\cal S}(\varepsilon,N)+\dots+a_d(\varepsilon,N){\cal S}(\varepsilon,
N+d)=h_t(N)\varepsilon^{t}+h_{t+1}(N)\varepsilon^{t+1}+\dots,
\end{equation}
for some functions $F_i(N)$, then the values of $F_i(N)$ can be computed by indefinite nested product-sum expressions $\tilde{F}_i(N)$. The depth of the $\tilde{F}_i(N)$ is $\leq\max_{t\leq j\leq i}(d_j+(i-j+1)d))$.
\end{theorem}

\begin{proof}
Choose $\mu\in\set N$ with $\mu\geq d$ such that $a_d(k)\neq0$ for all integers $k\geq\mu$ and such that the sequences $h_i(k)$ can be computed for indefinite nested product-sum expressions for each $k\geq\mu$.
Consider the $r$th iteration of the loop of Algorithm~\FLSR. Since $F_r(N)$ is a solution of~\eqref{KeyEquGeneric} with $X(N)=h_r(N)$ for all $N\geq\gamma$, $F_r(N)$ is a linear combination of~\eqref{Equ:SolutionSpace}. Taking the first $d$ initial values $F_r(\mu),\dots,F_r(\mu+d-1)$ the $c_i$ are uniquely determined. Induction on $r\in\set N$ proves the theorem. The bound on the depth is immediate.
\end{proof}

\noindent If the operator~\eqref{KeyEqu} factorizes as stated in~\eqref{FactorOp}, Alg.~\FLSR\ can be simplified as follows.

\smallskip

\noindent\textsf{Simplification~1.} The factorization~\eqref{FactorOp} needs to be computed only once and the solutions $F_i(N)$ can be obtained in terms of indefinite nested product-sum expressions by simply plugging in the results of the previous steps. E.g., for our running example, we get the generic solution

\vspace*{-0.8cm}

\begin{equation}\label{Equ:GenericSolSingleSum}
\frac{c_1}{N (N+2)}+c_2\frac{\displaystyle\sum_{i_1=1}^N \tfrac{-(-1)^{i_1}(2 i_1+1)}{i_1 \big(i_1+1\big)}}{2 N (N+2)}
-\frac{\displaystyle\sum_{i_1=1}^N\tfrac{(-1)^{i_1}(2 i_1+1)}{i_1 \big(i_1+1\big)}\sum_{i_2=1}^{i_1} \
\tfrac{(-1)^{i_2} i_2^2 X\big(i_2-2\big)}{\big(2 i_2-1\big) \big(2 i_2+1\big)}}{2 N (N+2)}
\end{equation}
of the recurrence
$a_0(0,N)F(N)+a_1(0,N) F(N+1)+a_2(0,N)F(N+2)=X(N)$
where the coefficients are defined as in~\eqref{Equ:SingleSumRecCoef}. In this way, one gets the solution $F_0(N)$ in terms of a double sum by setting $c_1=c_2=0$ and $X(i_2)=-24 i_2+48$ in~\eqref{Equ:GenericSolSingleSum}, i.e.,
\begin{equation}\label{Equ:NSSingleSumF0}
F_0(N)=\frac{-1}{2 N (N+2)}\sum_{i_1=1}^N \frac{(-1)^{i_1}(1+2 \
i_1)}{i_1(1+i_1)}\sum_{i_2=1}^{i_1}
\frac{-(-1)^{i_2}24i_2^3}{\big(-1+2 i_2\big) \big(1+2 i_2\big)}.
\end{equation}
One step further, one gets the solution $F_1(N)$ in terms of a quadruple sum by setting $c_1=c_2=0$ and plugging the double sum expression
$$X(i_2)=2i_2-20-\text{coeff}(a_0(\varepsilon,i_2)F_0(i_2)+a_1(\varepsilon,i_2)F_0(i_2+1)
+a_2(\varepsilon,i_2)F_0(i_2+2),\varepsilon)$$
into~\eqref{Equ:GenericSolSingleSum}.
Similarly, one obtains a sum expressions of $F_2(N)$ with nesting depth 6.

\smallskip

\noindent\textsf{Minimizing the nesting depth.} Given such highly nested sum expressions, the summation package \SigmaP\ finds alternative sum representations with minimal nesting depth. The underlying algorithms are based on a refined difference field theory worked out in~\cite{Schneider:08c,Schneider:11} that is adapted from Karr's $\Pi\Sigma$-fields originally introduced in~\cite{Karr}. E.g., with this machinery, we simplify the double sum~\eqref{Equ:NSSingleSumF0} to~\eqref{Equ:SingleSumF0}, and we reduce the quadruple sum expression for $F_1(N)$ to expressions in terms of single sums~\eqref{Equ:SingleSumF1}.

\smallskip

\noindent\textsf{Simplification 2:} The solutions~\eqref{Equ:HomSol} of the homogeneous version of the recurrence~\eqref{KeyEquGeneric} can be pre-simplified to expressions with minimal nesting depth by the algorithms mentioned above. Moreover, using the algorithmic theory described in~\cite{KS:06} the algorithms in~\cite{Schneider:08c} can be carried over to the sum expressions like~\eqref{Equ:InHomSol} involving an unspecified sequence $X(i_{d})$. With this machinery, \eqref{Equ:GenericSolSingleSum} simplifies to
\small
$$\frac{c_1}{N (N+2)}+\frac{c_2(-1)^{N+1}}{2 N (N+1) (N+2)}-\frac{\sum_{i_1=1}^N \tfrac{i_1 X(i_1-2)}{(2 i_1-1)
(2 i_1+1)}}{2 N (N+2)}-\frac{(-1)^N \sum_{i_1=1}^N \tfrac{(-1)^{i_1} \
i_1^2 X(i_1-2)}{(2 i_1-1)(2 i_1+1)}}{2 N (N+1) \
(N+2)}.
$$
\normalsize
Performing this extra simplification, the blow up of the nesting depth for the solutions $F_0(N),F_1(N),F_2(N),\dots$ reduces considerably: instead of nesting depth $2,4,6,\dots$ we get the nesting depths $1,2,3,\dots$. In particular, given these representations the simplification to expressions with optimal nesting depth in Step~2 also speeds up.

\smallskip

For simplicity we assumed that the $a_i(\varepsilon, N)$ are polynomials in $\varepsilon$. However, all arguments can be carried over immediately to the situation where the $a_i(\varepsilon, N)$ are formal power series
with the first coefficients given explicitly.
Moreover, our algorithm is applicable for more general sequences $a_i(N)$ and $h_i(N)$ whenever there are algorithms available that solve problem~\ProblemRS. E.g., if the coefficients $a_i(N)$ itself are expressible in terms of indefinite nested product-sum expression, problem~\ProblemRS\ can be solved by~\cite{ABPS:11}, and hence Algorithm~\FLSR\ is executable.

\subsection{An effective method for multi-sums}\label{Sec:EffectiveMethod}

For a multi-sum ${\cal S}(\varepsilon,N)$ with the properties (1)--(5) from  Assumption~\ref{Assum:SumProp} and with the assumption that it has a series expansion~\eqref{Equ:FExp2} for all $N\geq\lambda$ for some $\lambda\in\set N$,
the ideas of the previous section can be carried over as follows.

\smallskip

\noindent\textsf{{Step 1: Finding a recurrence.}} By WZ-theory~\cite[Cor.~3.3]{WZtheory} and ideas given in~\cite[Theorem~3.6]{Wegschaider} it is guaranteed that there is a recurrence
of the form
    \begin{equation}\label{Equ:HomRec}
    a_0(\varepsilon,N){\cal S}(\varepsilon,N)+\dots+a_d(\varepsilon,N){\cal S}(\varepsilon,N+d)=0
    \end{equation}
with coefficients $a_i(\varepsilon,N)\in\set K[\ep,N]$
    for the multi-sum ${\cal S}(\varepsilon,N)$ in $N$ that can be computed, e.g., by Wegschaider's algorithm; for infinite sums similar arguments have to be applied as in Step~2.2 of Section~\ref{Sec:FindRecurrences}. Given such a recurrence, let $\mu\in\set N$ with $\mu\geq\lambda$ such that $a_d(0,N)\neq0$ for all $N\in\set N$ with $N\geq\mu$.

\smallskip

\noindent\textsf{{Step 2: Determining initial values.}} If the sum~\eqref{physics-sums-general} contains no infinite sums, i.e., $s=0$, the initial values $F_i(k)$ in ${\cal S}(\varepsilon,k)=\sum_{i=t}^{\infty}F_i(k)\varepsilon^i$ for $k=\mu,\mu+1,\dots$
    can be computed immediately and can be expressed usually in terms of rational numbers. However, if infinite sums occur, it is not so obvious to which values these infinite sums evaluate for our general input class-- by assumption we only know that the $F_i(k)$ for a specific integer $k\geq\mu$ are real numbers. At this point we emphasize that our approach works regardless of whether we express these sums in terms of well known constants or we just keep the symbolic form in terms of infinite sums.
    In a nutshell, if we do not know how to represent these values in a better way, we keep the sum representation.
However, whenever possible it is desirable to rewrite these sums in terms of known values or special functions. Examples are harmonic sums which are known as limits for the external index $N \rightarrow \infty$, see~\cite{BK,Vermaseren:99},
to yield Euler-Zagier and multiple zeta values, cf.~\cite{MZV} and references therein, and generalized harmonic sums, see~\cite{Moch:02} which give special values of $S$-sums. In massive 2-loop computations and for the simpler 3-loop topologies these are
the only known classes, whereas extensions are known in case of more massive lines, cf.\ e.g.~\cite{BROAD}.

\smallskip

\noindent\textsf{{Step 3: Recurrence solving.}}
Given such a recurrence~\eqref{Equ:HomRec} together with the initial values of ${\cal S}(\varepsilon,N)$ (hopefully in a nice closed form) we can activate  Algorithm~\FLSR. Then by Corollary~\ref{Cor:ExpansionSolver}, we have a procedure that decides if the first coefficients of the expansion are expressible in terms of indefinite nested product-sum expressions.

\smallskip

Summarizing, we obtain Theorem~\ref{Thm:AlgForMultiSum} stated already in the beginning of this section. As mentioned already in the introduction, the proposed algorithm (see steps 1,2,3 from above) is not feasible for our examples arising form particle physics: forcing Weg\-schai\-der's implementation to find a homogeneous recurrence is extremely expensive
and usually fails due to the insufficient computational resources. Subsequently, we relax this restriction and search for recurrence relations which are not necessarily homogeneous.

\section{Finding recurrence relations for multi-sums}\label{Sec:FindRecurrences}
Given a multi-sum ${\cal S}(N)$ of the form (\ref{physics-sums-general}) we present a general method to compute a linear recurrence of ${\cal S}(N)$.
Here the challenge is to deal with infinite sums and summands which are not well defined outside the summation
range. We proceed as follows.

\medskip

\noindent\textsf{{Step~1: Finding a summand recurrence}.} The sum
(\ref{physics-sums-general}) fits the input class of the algorithm
\cite{Wegschaider}, an extension of multivariate WZ-summation due to
\cite{WZtheory}. This allows us to compute a recurrence for the
hypergeometric summand of~\eqref{physics-sums-general}.
Before giving further details, we recall that an expression
$\mathcal{F}\left( N,\sigma, j, \ep\right)$ is called hypergeometric in $N,\sigma,j$,
if there are rational functions $r_{\nu,\mu,\eta}(N,\sigma, j, \ep)\in\set K(N,\sigma,j,\ep)$
such that
$\frac{\mathcal{F}(N,\sigma, j, \ep)}{\mathcal{F}(N+\nu,\sigma+\mu,
j+\eta, \ep)}=
r_{\nu,\mu,\eta}(N,\sigma, j, \ep)$ at the points $(\nu,\mu,\eta) \in
\mathbb{Z}^{r+s+2}$  where this ratio is defined.
Then the Mathematica package {\tt MultiSum} described in \cite{Wegschaider} solves the following problem by coefficient comparison and solving the underlying system of linear equations.

\smallskip

\noindent \textbf{Given} a hypergeometric term $\mathcal{F} \left( N, \sigma, j, \ep\right)$, a finite structure set $\mathbb{S} \subset \mathbb{N}^{s+r+2}$ (w.l.o.g.\ we restrict to positive shifts) and degree bounds $B\in\set N$, $\beta\in\set N^s$, $b\in\set N^{r+1}$.\\
\textbf{Find}, if possible, a recurrence of the form
\begin{equation}\label{non-kfree-Feynman}
\sum\limits_{\left( u,v, w\right)\in
\mathbb{S}}
c_{u,v,w}\left( N, \sigma, j, \ep\right) \mathcal{F} \left( N+u, \sigma
+v, j+w, \ep\right) =0
\end{equation}
with polynomial coefficients $c_{u,v,w}\in\set K[N,\sigma,j,\ep]$, not all zero, where the degrees of the variables $N$, $j_i$ and $\sigma_i$ are bounded by $B$, $\beta_i$ and $b_i$, respectively.

\smallskip

\begin{remark}\label{remark:AlgorithmSetUp}
\textbf{(1)}
In general, choosing $\set S$ large enough, there always exists a summand recurrence~\eqref{non-kfree-Feynman} for proper hypergeometric summands $\mathcal{F}$ (see Footnote~\ref{Footnote:Proper}) due to~\cite{WZtheory}.
In all our computations we found such a recurrence by setting the degree bounds to $1$, i.e., $B=\beta_i=b_i=1$.\\
\textbf{(2)} To determine a small structure set $\mathbb{S}
\subseteq \mathbb{N}^{s+r+2}$ which provides a solution w.r.t.\ our fixed degree bounds, A.~Riese and B.~Zimmermann enhanced the package \texttt{MultiSum} by a method based on modular computations. In this way one can loop through possible choices inexpensively until one succeeds to find such a recurrence~\eqref{non-kfree-Feynman}.
\end{remark}

\vspace*{-0.2cm}

Next, the algorithm successively divides the polynomial
recurrence operator  (\ref{non-kfree-Feynman}) by all forward-shift
difference operators
\begin{equation*}\label{define-Delta}
\Delta_{\sigma_i}\mathcal{F}(N, \sigma, j, \ep):= \mathcal{F}\left(
N,\sigma_1,\dots,\sigma_i+1, \dots ,\sigma_s, j,\ep\right) -
\mathcal{F}(N, \sigma, j, \ep)
\end{equation*}
for $ 1 \leq i \leq s$, as well as by similar $\Delta$-operators defined for
the variables from $j_i$ which have finite summation bounds.

At last we obtain an operator free of shifts in the summation variables
$(\sigma, j)$ called the principal part of the recurrence
(\ref{non-kfree-Feynman}) which equals the sum of all delta parts in the
summation variables from $(\sigma, j)$, i.e.,
\begin{alignat}1\nonumber
\sum\limits_{m \in \mathbb{S}' }   & a_m (\ep, N)
\mathcal{F} (N+m, \sigma, j, \ep) =\sum\limits_{l=0}^{r} \Delta_{j_l} \bigg(
\sum\limits_{( m,n) \in \mathbb{S}'_{l}} d_{m,n} ( N,\sigma,j,\ep)
\mathcal{F} (N+m, \sigma, j+n, \ep)\bigg)
\\[-0.2cm]
 &+ \sum\limits_{i=1}^s \Delta_{\sigma_i} \bigg(
\sum\limits_{( m,k,n)\in \mathbb{S}_i} b_{m,k,n} ( N,\sigma,j,\ep)
\mathcal{F}(N+m, \sigma+k, j+n, \ep) \bigg)\label{certific-Feynman}
\end{alignat}
where the coefficients $a_m$, usually not all zero (see Remark~\ref{Remark:DegreeBounds}.2), $b_{m,k,n}$ and $d_{m,n}$ are
polynomials and the sets $\set S'\subset\set N$, $\set S_i\subset\set N^{s+r+2}$ and $\set S'_l\subset\set N^{r+2}$ are finite.
Recurrences of the form  (\ref{certific-Feynman}) satisfied by the
hypergeometric summand are called
certificate recurrences and have polynomial
coefficients $a_{m}\left(\ep, N\right)$ free of the summation variables
from $(\sigma, j)$, while the coefficients of the delta-parts are polynomials
involving all variables.

\begin{remark}\label{Remark:DegreeBounds}
\textbf{(1)} In principle, the degrees of the polynomials $b_{m,k,n}$ and $d_{m,n}$ arising in~\eqref{certific-Feynman} can be chosen arbitrarily large w.r.t.\ $\sigma_i$ and $j_i$. However, in Step~2 we will sum~\eqref{certific-Feynman} over the input range and hence we have to guarantee that the resulting sums over~\eqref{certific-Feynman} are well defined. As a consequence, the degrees of the $d_{m,n}$ and $b_{m,k,n}$ w.r.t.\ the variables $\sigma_i$ have to be chosen carefully if in (\ref{limitconsiderations}) one of the constants $c_i$ is zero. As mentioned earlier, for such situations we restrict ourselves to the case $s=1$. In this case, the degree in the $b_{m,k,n}$ should be smaller than the constant $d_1$ from~(\ref{limitconsiderations}) and the degree in the $d_{m,n}$ should be not bigger than the constant $o$ from~\eqref{Equ:SumDefined}. To control this total bound $b:=\min(d_1-1,o)$, we exploit the following observation~\cite[p.~43]{Wegschaider}:
While transforming~\eqref{non-kfree-Feynman} to~\eqref{certific-Feynman} by dividing through the operators~\eqref{define-Delta}, one only has to perform a simple sequence of additions of the occurring coefficients in~\eqref{non-kfree-Feynman}, and thus the degrees w.r.t.\ the variables do not increase. Summarizing, if we choose $\beta_1$ in our ansatz such that $\beta_1<b$, the degrees in the $b_{m,k,n}$ and $d_{m,n}$ w.r.t.\ the variable $\sigma_1$ are smaller than $b$.\\
\textbf{(2)} In general, it might happen that the principal part is $0$, i.e., we get a trivial remainder within the operator divisions.
In~\cite[Thm.~3.2]{Wegschaider} this situation was resolved at the
cost of increasing the degrees w.r.t. some of the variables.
If within this construction the degree w.r.t.\ $\sigma_1$ increases too much, manual adjustment is needed (e.g., force the structure set to be different or change the degree bounds manually). However, this exotic case never occurred within our computations.
\end{remark}

\begin{example}
For the sum
\small
\begin{equation} \label{Ex-75}
 \mathcal{S}\left(\ep,N\right):=
\sum\limits_{j_0=0}^{N-3}
\sum\limits_{j_1 =0}^{N-3-j_0}
\underbrace{(-1)^{j_1} (j_1+1)\binom{N-2-j_0}{j_1+1}
 \frac{\Gamma (j_0+j_1+1)
    \left(1-\frac{\ep}{2}\right)_{j_0}
    \left(3-\frac{\ep}{2}\right)_{j_1}}{
(4-\ep)_{j_0+j_1} \left(\frac{\ep}{2}+4\right)_{j_0+j_1}}}_{=:\mathcal{F}(N,j_0,j_1)}
\end{equation}
\normalsize
with the discrete parameter $N \geq 3$ and $\ep >0$ the package \texttt{MultiSum} computes the summand recurrence
\begin{multline}\label{Equ:Certificate}
(\ep -2 N) N {\cal F}(N,j_0,j_1)-(\ep -N-3) (\ep +2
    N+2) {\cal F}(N+1,j_0,j_1)\\
= \Delta_{j_0} [(\ep
    ^2+j_0 \ep +\ep -2 j_1-2 j_0 N-4 j_1
    N-12 N-6) {\cal F}(N+1,j_0,j_1)]\\
    +\Delta_{j_1}[(\ep -2 N)(j_0+j_1-N+1)
    {\cal F}(N,j_0,j_1)\\
    +(-2 N^2+\ep  N+2 j_0 N+4
    j_1 N+4 N-2 \ep -\ep  j_0+2 j_1)
    {\cal F}(N+1,j_0,j_1]).
\end{multline}
\end{example}

\medskip

\noindent\textsf{{Step~2: A recurrence for the sum}}.
Taking as input the certificate recurrences (\ref{certific-Feynman}) we algorithmically find
the inhomogeneous part of the recurrence satisfied by the sum
(\ref{physics-sums-general}) which will contain special instances of the
original multi-sum of lower nesting
depth.
\smallskip

The recurrence for the multi-sum
(\ref{physics-sums-general}) is obtained by summing the certificate recurrence
(\ref{certific-Feynman}) over all variables from $(\sigma, j)$ in the given
summation range $\mathcal{R} \subseteq
\mathbb{Z}^{s+r+1}$.
Since it can be easily checked whether the
summand $\mathcal{F}$ satisfies the~(\ref{certific-Feynman}), the certificate recurrence also provides an
algorithmic proof of the recurrence for the multi-sum $\mathcal{S}(N,\ep)$. In particular, since we set up the degrees of the coefficients in~(\ref{certific-Feynman}) w.r.t.\ the variables accordingly, see Remark~\ref{Remark:DegreeBounds}, it follows that the resulting sums are analytically well defined.

To pass from the certificate recurrence to a homogeneous or inhomogeneous
recurrences for the sum, special emphasis has to be put on the
$\Delta$-operators. In particular, the finite summation bounds appearing in
(\ref{physics-sums-general}) lead to an inhomogeneous right hand side after
summing over the summand recurrence~\eqref{certific-Feynman}.
A method to set up the inhomogeneous
recurrences for the summation problems (\ref{physics-sums-general}) was
introduced in \cite[Chapter 3]{Stan-thesis}.
We summarize the steps of this approach implemented in
the package \texttt{FSums}.

\noindent In this context, we use tuples to denote multi-dimensional
intervals. The range represented by the tuple interval $[i,k]$ is the
Cartesian product of the intervals defined by the components $i,k \in
\mathbb{Z}^n$. More
precisely,
$[i,k]:= [i_1,k_1] \times [i_2,k_2] \times \cdots
	\times [i_n ,k_n]$ where $[i_j,k_j]=\{i_j,i_j+1,\dots,k_j\}$.
Often when working with nested sums, summation ranges for inner sums will
depend on the value of a variable for an outer sum. Intervals whose endpoints
are defined by tuples are not enough to represent the summation ranges for
these sums. We will use a variant of the cartesian product notation to
denote such a summation range. Namely, to refer to
a variable associated to a range, we will specify it as a subscript at the
corresponding interval and use $\ltimes$ signs instead of the $\times$ symbols.
For example, the range for the
sum (\ref{Ex-857}) can be written as
$[0,\infty) \times [0,N-3]_{j_0} \ltimes [0,N-j_0-3] \ltimes [0,j_0 +1].$
We also introduce this notation for
the initial range of the sum (\ref{physics-sums-general}) as
\begin{equation}\label{input-range}
 \mathcal{R} := \mathcal{R}_\sigma \times \mathcal{R}_j
\end{equation}
\noindent where $
\mathcal{R}_\sigma := [p,\infty)$ and
$
\mathcal{R}_j  = [q_0,L_1(N)] \ltimes \dots
\ltimes [q_r,L_r(N,j_0,\dots,j_{r-1})]$, are the infinite and the finite
range, respectively.

\noindent\textsf{Step~2.1: Refining the input sum.}
As indicated earlier, we consider the summands from (\ref{physics-sums-general}) as
well-defined only inside the
initial input range
$
 \mathcal{R} \subseteq
D_\mathcal{F}
$
where $D_\mathcal{F}$ denotes the set of well-defined values for the
proper hypergeometric function $\mathcal{F}$. Because of this restriction we
need to determine a possible smaller summation range over which we are allowed
to sum the certificate recurrences (\ref{certific-Feynman}).

\vspace*{-0.2cm}

\begin{example}
We illustrate this phenomenon by our concrete example~(\ref{Ex-75}).
Let us start by summing over the initial summation range
$\mathcal{R} = [0,N-3]_{j_0} \ltimes [0,N-3-j_0]$
over the delta parts on the right hand side of the recurrence~\eqref{Equ:Certificate}
which is of the form (\ref{certific-Feynman}).
For this we denote the polynomial  coefficients inside the delta parts
$\Delta_{j_0}$ and $\Delta_{j_1}$ with $ e(N,j_0,j_1, \ep)$ and $
d_1(N,j_0,j_1,\ep)$,
$d_2(N,j_0,j_1, \ep)$, respectively.
By summing over the first term inside the $\Delta_{j_1}$-part and using the
telescoping property, we have

\vspace*{-0.5cm}

\begin{alignat*}1
\sum\limits_{j_0=0}^{N-3} &\sum\limits_{j_1 =0}^{N-3-j_0} 	
\Delta_{j_1} [ d_1(N,j_0,j_1, \ep) {\cal F}(N,j_0,j_1)]
 =
\sum\limits_{j_0 =0}^{N-3}
\left( d_1(N,j_0,j_1, \ep) {\cal F}(N,j_0,j_1)\right)
\Big|_{j_1=0}^{j_1=N-2-j_0}
\\[-0.2cm] &=
\sum\limits_{j_0 =0}^{N-3}  d_1(N,j_0,N-2-j_0, \ep) {\cal F}(N,j_0,N-2-j_0) -
\sum\limits_{j_0 =0}^{N-3}  d_1(N,j_0,0, \ep) {\cal F}(N,j_0,0)
\end{alignat*}

\vspace*{-0.2cm}

\noindent where we use the short-hand notation
$\sum_{k=0}^{l}\mathcal{F} ( k, l )
\big|_{l=A}^{l=B} := \sum_{k=0}^{B} \mathcal{F}(k, B)
-\sum_{k=0}^{A} \mathcal{F}(k, A)$.
We observe that, after telescoping, the upper bound $N-2-j_0$ for $j_1$
translates into a term outside the original summation range. To work under
the assumption that our summand ${\cal F}(N,j_0,j_1)$ is well-defined only inside its
range $\mathcal{R}$, we need to adjust the range over which we sum the
certificate recurrence or shift this relation with respect to the free parameter
$N$. As discussed in \cite[Chapter 3]{Stan-thesis}, the approach based on
computing a smaller admissible summation range is more efficient since it leads
to fewer new sums in the inhomogeneous parts of the recurrences.\\
In the case of our example $\mathcal{S}(\ep,N)$, we consider the
new range
$\mathcal{R}' =  [0, N-4]_{j_0}\ltimes [0,N-j_0-4].$
As a consequence we compute separately a single sum which was called
in \cite[Chapter 3]{Stan-thesis} a sore spot,

\vspace*{-0.3cm}

\begin{equation}\label{output-GenSoreSpots}
\mathcal{S}(\ep,N) = \sum\limits_{j_0=0}^{N-4} \sum\limits_{j_1
=0}^{N-4-j_0} 	 {\cal F}(N,j_0,j_1)  + \sum\limits_{j_0=0}^{N-3}   {\cal F}(N,j_0,N-j_0-3).
\end{equation}
\end{example}

In general, the package {\tt FSums} contains an algorithm that determines the
inevitable summation range and computes the necessary
sore spots for sums of the form (\ref{physics-sums-general}); these extra sums with lower nesting depth have to be considered separately (see also the DIVIDE step in our method described in Section~\ref{Sec:EfficientApproach}).
Subsequently, we denote the sum over the restricted range $\mathcal{R}'$ by
$\mathcal{S}'( \ep,N)$.

\medskip

\noindent\textsf{Step~2.2: Determining the inhomogeneous part of the recurrence.}
Summing a certificate recurrence of the form (\ref{certific-Feynman}) over
the restricted range $\mathcal{R}'$ determined in the previous step leads to a
recurrence for the new sum $\mathcal{S}'(\ep,N)$. The inhomogeneous part
contains special instances of this sum of lower nesting depth. Next, we introduce the
types of sums appearing on the right hand side.

\smallskip

\noindent\textsf{Step~2.2.1: The finite summation bounds.}
Shift compensating sums are the first side-effect of nonstandard summation
bounds. They appear when we sum over the left hand side of the recurrence over
a given definite range, because our upper summation bounds depend on the other
summation parameters.

\begin{example}
Subsequently, we will illustrate these aspects with our running example (\ref{Ex-75}). As deduced from Step~2.1, we continue from now on
with the new sum

\vspace*{-0.3cm}

\begin{equation}\label{Equ:NewSum}
 \mathcal{S}'(\ep,N) = \sum\limits_{j_0=0}^{N-4} \sum\limits_{j_1
=0}^{N-4-j_0} 	 {\cal F}(N,j_0,j_1).
\end{equation}

\vspace*{-0.cm}

\noindent When we sum the certificate recurrence~\eqref{Equ:Certificate} over
the restricted range $\mathcal{R}'$, we obtain

\vspace*{-0.4cm}

\begin{equation}\label{output-compensate-terms}	
\sum\limits_{j_0=0}^{N-4} \sum\limits_{j_1 =0}^{N-4-j_0}
{\cal F}(N+1,j_0,j_1) = \mathcal{S}'(\ep,N+1) - \sum\limits_{j =0}^{N-3}
{\cal F}(N+1,j,N-3-j).
\end{equation}
\end{example}

Compensating sums of this form  appear
only in the case of upper summation bounds depending on the free variable $N$.
After summing over the left hand side
of the recurrence, we will move
the resulting compensating sums, with a change of sign, to the inhomogeneous
part.

\begin{example}
Including the new shifted sum as the
first term of the output, the following
procedure of \texttt{FSum} delivers the right hand side of (\ref{output-compensate-terms})
\begin{mma}
\In |ShiftCompensatingSums|[F[N,j_0, j_1],
\{\{j_0,0,N-4\},\{j_1,0,N-4-j_0\}\},N,1]\\

\vspace*{-0.2cm}

\Out \text{SUM}[N+1] + \text{FSum}[ -F[1 + N, j_0, -3 - j_0 + N], \{\{j_0, 0,
-3 + N\}\}].\\
\end{mma}
\noindent Note that we use the structure {\tt FSum} to store sums with
nonstandard
boundary conditions of the form (\ref{physics-sums-general}). This data type
contains two components, the summand and a list structure for the
summation range. The nested range is stored in the order given in
(\ref{physics-sums-general}), starting with the infinite sums and ending with
the sums with finite summation bounds in the order of their dependence.
\end{example}

When summing over the $\Delta$-parts we generate two types
of sums on the right side of the recurrence, the $\Delta$-boundary sums and
the so-called telescoping compensating sums.

\begin{example}
When summing over the
$\Delta_{j_0}$-part of the recurrence~\eqref{Equ:Certificate}, we get

\vspace*{-0.5cm}
\small
\begin{multline*}
\sum\limits_{j_0=0}^{N-3}
\sum\limits_{j_1 =0}^{N-3-j_0} 	
\Delta _{j_0} [ e(N,j_0,j_1, \ep) {\cal F}(N+1,j_0,j_1)]\\
=\sum\limits_{j_0=1}^{N-2} \sum\limits_{j_1 =0}^{N-2-j_0}  	
 e(N,j_0,j_1, \epsilon) F(N+1,j_0,j_1)-\sum\limits_{j_0=0}^{N-3} \sum\limits_{j_1 =0}^{N-3-j_0} 	
 e(N,j_0,j_1, \epsilon) F(N+1,j_0,j_1).
\end{multline*}

\vspace*{-0.cm}

\normalsize
\noindent Now one sees that exactly the sum with the summation index $j_0$ cancels and one obtains

\vspace*{-0.4cm}

\small
$$\left. \sum\limits_{j_1 =0}^{N-3-j_0}
( e(N,j_0,j_1, \ep) {\cal F}(N+1,j_0,j_1)) \right|_{j_0=0}^{j_0=N-2}+
\sum\limits_{j_0=1}^{N-2} e(N,j_0,N-2-j_0, \ep) {\cal F}(N+1,j_0,N-2-j_0).$$

\normalsize
\end{example}

\vspace*{-0.2cm}

Because of the structure of the summation bounds for the nested sums
(\ref{physics-sums-general}) we can use again our procedure \texttt{ShiftCompensatingSums} to generate the shift compensating sums and to read off the telescoping compensating sums. This connection becomes
clearer when we consider the more involved sum
(\ref{Ex-857}) (with its restricted range $N-4$ instead of its original range $N-3$) and apply, e.g., the $\Delta_{j_0}$-operator:

\vspace*{-0.6cm}

\begin{alignat}1\label{ex-shift-comp-terms}\nonumber
\sum\limits_{ \sigma_0 =0}^{\infty}
&\sum\limits_{ j_0=0}^{N-4}
\sum\limits_{ j_1=0}^{N-j_0-4}
\sum\limits_{ j_2=0}^{j_0} \Delta_{j_0}\left[{\cal F}(N, \sigma_0, j_0,
j_1, j_2) \right]
=  \left. \sum\limits_{ \sigma_0 =0}^{\infty}
\sum\limits_{ j_1=0}^{N-j_0-4}
\sum\limits_{ j_2=0}^{j_0} {\cal F}(N, \sigma_0, j_0,
j_1, j_2) \right|_{j_0=0}^{j_0=N-3}
\\[-0.2cm] &\nonumber
+ \sum\limits_{\sigma_0 =0}^{\infty}
\sum\limits_{j_0=1}^{N-3}
\sum\limits_{j_2=0}^{j_0-1}{\cal F}(N, \sigma_0, j_0,
N-j_0-3, j_2)
- \sum\limits_{\sigma_0 =0}^{\infty}
\sum\limits_{j_0=1}^{N-3}
\sum\limits_{j_1=0}^{N-j_0-4}{\cal F}(N, \sigma_0, j_0,
j_1, j_0);
\end{alignat}
note that the first element on the right side of this identity produces the
$\Delta$-boundary sums while the last two are due to telescoping
compensation. More precisely, with
\begin{mma}
\In |ShiftCompensatingSums|[ F[N, \sigma_0, j_0 -1, j_1,j_2],
\{ \{ \sigma_0, 0, \infty \}, \{ j_1, 0, N-j_0-4 \}, \hfill \break \{ j_2,0,j_0
\} \} /. j_0
\rightarrow (j_0 -1), j_0, 1 ]\\

\vspace*{-0.2cm}

\Out
\{
|FSum| [ F[N, \sigma_0, j_0, j_1, j_2],
 \{ \{\sigma_0, 0, \infty\}, \{j_1, 0, N-4 - j_0 \}, \{j_2, 0,
j_0\}\}],
|FSum| [F[N, \sigma_0, j_0, N-3 - j_0 , j2],
   \{\{\sigma_0, 0, \infty\}, \{j_2, 0,  j_0-1\}\}],
|FSum| [-F[N, \sigma_0, j_0, j_1, j_0],
    \{\{\sigma_0, 0, \infty\}, \{j_1, 0, N-4 - j_0 \}\}]
\} \\
\end{mma}
\noindent we obtain exactly this result: the delta boundary sums are obtained by evaluating the
first entry of the output for $j_0=0$ and $j_0 = N-3$ and the compensating sums result by
adding the shifted sum $[1,N-3]_{j_0}$ to the range of the other terms in
the output. A detailed description of these computations can be found in~\cite[Alg.~4]{Stan-thesis}.

\smallskip

\noindent\textsf{Step~2.2.2: The infinite summation bounds.}
To sum over the delta parts in~\eqref{certific-Feynman} coming from the summation variables
$\sigma_i$, e.g.,
$\Delta_{\sigma_i} b_{m,k,n}( N,\sigma,j,\ep)
\mathcal{F}(N+m, \sigma+k, j+n, \ep)$ we have to ensure that
$\lim_{\sigma_i\to\infty}b_{m,k,n}( N,\sigma,j,\ep)
\mathcal{F}(N+m, \sigma+k, j+n, \ep)$
exists. Looking at the asymptotic conditions (\ref{limitconsiderations}) of the input sum (\ref{physics-sums-general}), there will be no problem if $c_i>0$. However, if the constant $c_i$ is zero, we need to verify that the degrees
of the polynomial coefficients $b_{m,k,n}$ appearing in the respective
$\Delta_{\sigma_i}$-part are smaller than the bound $\beta_i$. As worked out in Remark~\ref{Remark:DegreeBounds} this property is guaranteed by our ansatz.

\medskip

The above sections introduced the types of sums, i.e., shift and telescoping
compensating sums as well as delta boundary sums,  which will appear on the
right hand side of the inhomogeneous recurrences satisfied by summation
problems of the form (\ref{physics-sums-general}) after summing over
corresponding certificate recurrences (\ref{certific-Feynman}).
A procedure to generate these inhomogeneous recurrences is implemented in the
package {\tt FSums}. E.g., the recurrence satisfied by
the sum $\mathcal{S}'(\ep,N)$, which we denote by $\text{SUM}[N]$,
is returned by
\begin{mma}
\In
|finalRecS| = |InhomogenRec|[|certRecS|, \{\{j_0, 0, -4 + N \},
\{ j_1, 0, -4 - j_0 + N \}\}, N ]\label{MMA:InhomRec}\\

\vspace*{-0.2cm}

\Out
(\ep - 2 N) N |SUM|[N] + (3 - \ep + N) (2 + \ep + 2 N)
   |SUM|[1 + N] ==
\linebreak |FSum|[ (1 + j_0 - N) (-\ep + 2 N) F[N, j_0,
0],\{ \{j_0, 0, -4 + N\} \}] +
\linebreak|FSum|[-2 (\ep - 2 N) F[N, j_0, -3 - j_0 +
N],  \{\{j_0, 0, -4 + N\}\}] +
\linebreak |FSum|[(\ep - 2 N) (2 + j_0 - N) F[1 + N, j_0, 0], \{\{j_0, 0,
-4 +N\}\}] +
\linebreak |FSum|[(6 - \ep - \ep^2 + 2 j_1 + 12N
+ 4 j_1 N) F[1 + N, 0, j_1], \{\{j_1, 0, -4 + N\}\}] +
\linebreak |FSum|[(3 - \ep + N) (2 + \ep + 2 N) F[1 + N, j_0, -3 -
j_0 + N], \{ \{j_0, 0, -3 + N\}\}] +
\linebreak|FSum|[(\ep + \ep^2 + 2 j_0 +
\ep j_0 -  2 N + 2 j_0 N - 4 N^2) F[1 + N, j_0, -3 - j_0 + N], \{\{j_0, 1,
-3 + N\}\}] +
\linebreak |FSum|[-((6 + 2 \ep + 2 j_0 + \ep j_0 + 6 N - \ep
N + 2 j_0 N - 2N^2) F[1 + N, j_0, -3 - j_0 + N]), \{\{j_0, 0, -4 + N \}\} ]; \\
\end{mma}

\vspace*{-0.2cm}

\noindent here \texttt{certRecS} stands for the certificate recurrence~\eqref{Equ:Certificate}.

\section{An efficient approach to find $\varepsilon$-expansions for multi-sums}\label{Sec:EfficientApproach}

Let ${\cal S}(\varepsilon,N)$ be a multi-sum of the form~\eqref{physics-sums-general} with the properties (1)--(5) from  Assumption~\ref{Assum:SumProp} and assume that ${\cal S}(\varepsilon,N)$ has a series expansion~\eqref{Equ:FExp2} for all $N\geq\lambda$ for some $\lambda\in\set N$. Combining the methods of the previous sections we obtain the following general method to compute the first coefficients, say $F_t(N),\dots,F_u(N)$ of~\eqref{Equ:FExp2}.

\medskip

\noindent\textbf{Divide and conquer strategy}

\begin{enumerate}
\item BASE CASE: If $\mathcal{S}(\ep, N)$ has no summation quantifiers, compute the expansion by formulas such as~\eqref{eq:A6} and~\eqref{Equ:Expan}.

\item DIVIDE: As worked out in Section~\ref{Sec:FindRecurrences}, compute a recurrence relation
\begin{equation}\label{Equ:SumRecurrence}
a_0(\varepsilon,N)\mathcal{S}(\ep, N)+\dots+a_d(\varepsilon,N)\mathcal{S}(\ep, N+d)=h(\varepsilon,N)
\end{equation}
with polynomial coefficients
$a_i(\varepsilon,N)\in\set K[\ep,N]$, $a_m(\ep,N)\neq0$ and the right side $h(\varepsilon,N)$ containing a linear
combination of hypergeometric multi-sums each with less than $s+r+1$ summation
quantifiers. Note: In some cases, the sum has to be refined and some ``sore spots'' (again with fewer summation quantifiers) have to be treated separately by calling our method again; see Step~2.1 in Section~\ref{Sec:FindRecurrences}.

\item CONQUER: Apply the strategy recursively to the simpler sums in
$h(\varepsilon,N)$. This results in an expansion of the form
\begin{equation}\label{Equ:hExpansion}
h(\varepsilon,
N)=h_t(N)\varepsilon^{t}+h_{t+1}(N)\varepsilon^{t+1}+\dots+h_u(N)\varepsilon^u+O(\varepsilon^{u+1});
\end{equation}
if the method fails to find the $h_t(N),\dots,h_u(N)$ in terms of indefinite
nested product-sum expressions, STOP.

\item COMBINE: Given~\eqref{Equ:SumRecurrence} with\footnote{Cf.\ Step~2 of Section~\ref{Sec:EffectiveMethod} to see how we deal with the initial values.}~\eqref{Equ:hExpansion},
compute, if possible, the $F_t(N),\dots,F_u(N)$ of~\eqref{Equ:FExp2} in terms of
nested product-sum expressions by executing Algorithm~\FLSR.
\end{enumerate}
\normalsize
\medskip

We illustrate our method with the double sum~\eqref{Ex-75}; internally we transform all the objects in terms of $\Gamma(x)$-functions in order to apply expansion formulas such as~\eqref{eq:A6} and~\eqref{Equ:Expan}. First, we compute the summand recurrence given in~\eqref{Equ:Certificate}. While computing a recurrence for the sum itself, it turns out that we have to refine the summation range, i.e., our computation splits into two problems as given in~\eqref{output-GenSoreSpots}. We continue with the refined double sum~\eqref{Equ:NewSum} and obtain the inhomogeneous recurrence \texttt{finalRecS} given in {\sffamily Out[\ref{MMA:InhomRec}]}. Now we apply recursively our method and compute successively expansions for each of the single sums on the right hand side; see also Example~\ref{Exp:SingleExt1}. Adding all the expansions termwise gives the recurrence
\small
\begin{alignat*}1
(\ep& -2 N) N \mathcal{S}'(\ep,N)-(\ep -N-3) (\ep +2 N+2)
    \mathcal{S}'(\ep,N+1) =
\\&
\frac{18 (2 N^6-3 N^5-8 N^4+13 N^3-4 N+8)}{(N-2) (N-1) N (N+1)
    (N+2)}-\frac{36 (2 N^4+N^3-9 N^2-2 N+4)(-1)^N}{(N-2) (N-1) N (N+1) (N+2)}
\\& + \ep \left[
 \frac{3 (N^8-6 N^7-32 N^6+20 N^5+151 N^4+14 N^3-200 N^2-28
    N+56)}{(N-2) (N-1) N (N+1)^2 (N+2)^2} \right.
\\& + \left. \frac{6 (2 N^6+N^5-14 N^4+9 N^3+40 N^2-22 N-28) (-1)^N}{(N-2) (N-1)
N (N+1)^2 (N+2)^2}+\frac{36 S_1(N)}{N+1}
\right]
\\& + \ep^2 \left[
  \frac{9 S_1(N)^2}{N+1}-\frac{6 (N-5) S_1(N)}{(N+1)^2}
-\frac{N^6 (5
N^3+48 N^2+246 N+568)}{4 (N-1)(N-2)(N+1)^3(N+2)^3}
\right.
\\ &\left(\frac{9 \left(N^4-N^3-4 N^2+4 N+8\right)}{(N-2) (N-1) N
    (N+1) (N+2)}-\frac{18 \left(2 N^4+N^3-9 N^2-2 N+4\right)
    (-1)^N}{(N-2) (N-1) N (N+1) (N+2)}\right)
    S_2(N)
\\& + \left. + \frac{363 N^6+3720 N^5+3672 N^4-5280 N^3-10712 N^2-4592 N-128}{4 N
    (N-1)(N-2)(N+1)^3(N+2)^3}  \right]+O(\ep^3).
\end{alignat*}
\normalsize
Together with its first initial value
$\mathcal{S}'(\ep,4)=\frac{27}{16} -\frac{1}{128}\ep
-\frac{11}{1024}\ep^2$
Algorithm~\FLSR\ computes the series expansion of $\mathcal{S}'(\ep,N)$.
Finally, we compute the expansion of the extra sum $\sum_{j_0=0}^{N-3}{\cal F}(N, j_0, N-3-j_0)$ with our method, and adding this result to our previous computation leads to the final result
\small
\begin{alignat*}1
\mathcal{S}&(\ep,N) = \frac{81 (N^2-3 N+2)}{4 N^2}+ \ep \left[
   \frac{3 (N^4-13 N^3-28 N^2-32 N+24)}{8 N^3 (N+2)}+\frac{9 (N+3)
    S_1(N)}{N (N+1)(N+2)} \right]
\\& + \ep^2 \left[
  \frac{9 (N+3) S_1(N)^2}{4 N (N+1) (N+2)}-\frac{3(5 N^3+36 N^2+37
    N-18) S_1(N)}{4 N (N+1)^2 (N+2)^2} \right. +\frac{9 (N^2+3 N+4)
    S_2(N)}{4 N^2 (N+1) (N+2)}
\\ &- \left.\frac{5 N^6+17 N^5+162 N^4+208
    N^3+592 N^2+240 N-288}{32 N^4 (N+2)^2}
\right]+O(\ep^3).
\end{alignat*}
\normalsize
Similarly, we compute, e.g., the first two coefficients of the expansion of the sum~\eqref{Ex-857}:
\small
\begin{align*}
\mathcal{U}&(\ep,N)=
 \frac{3 (-1)^N \big(N^2+2 N-1\big) S_1(N)}{N (N+1)}-9 (-1)^N+\frac{6 \
(-1)^N S_2(N)}{N}+ \\[-0.1cm]
&\ep\Big[\zeta(2) \big(-\frac{3 (-1)^N (4 N+3)}{2 N}-\frac{3 (-1)^N (3 N+2) \
S_{-1}(N)}{N}+\frac{9}{2 N}\big)+\frac{3 (-1)^N S_1(N)^2}{2 \
(N+1)}\\[-0.1cm]
&+\frac{(-1)^N \big(2 N^4+34 N^3+101 N^2+89 N+2\big) S_1(N)}{2 N (N+1)^2 \
(N+2)}+\frac{3 (-1)^N \big(4 N^2+14 N+13\big)}{(N+1) (N+2)}\\[-0.1cm]
&+\frac{(-1)^N \
\big(-30 N^2-38 N+1\big) S_2(N)}{2 N (N+1)}-\frac{9 (-1)^N (2 N+1)
S_3(N)}{N}+\frac{9 \
S_{-2}(N)}{N}+9 (-1)^N S_{2,1}(N)\\[-0.1cm]
&-\frac{6 (-1)^N (3 N+2) S_{-2}(N) S_{-1}(N)}{N}+\frac{6 (-1)^N (3 N+2) \
S_{-2,-1}(N)}{N}\Big]+O(\ep^2)
\end{align*}
\normalsize
where $\zeta(2)=\sum_{i=1}^\infty\frac{1}{i^2}=\pi^2/6$.

\begin{remark}\label{Remark:MethodConclusion}
In the following we give further comments on our proposed method and provide strategies for using it in the context of the evaluation of Feynman integrals.\\
\textsf{1. A heuristic.} The conquer step turns our procedure into a method and not into an algorithm.
Knowing that there is an expansion of ${\cal S}(\ep,N)$ in terms of indefinite nested sums and products
and plugging this solution into the left hand side of~\eqref{Equ:SumRecurrence} shows that also the right
hand side of~\eqref{Equ:SumRecurrence} can be written in terms of indefinite nested product-sum expressions.
But in our method the right hand side is split into various sub-sums and it is not guaranteed that each
sum on its own is expressible
in terms of indefinite nested product-sum expressions -- only the combination has this particular form.
However, for our input class arising from Feynman-integrals this method always worked.\\
\textsf{2. A hybrid version for speed--ups.} As it turned out,
the bottleneck in our computations is the task to compute a recurrence of the
form~\eqref{Equ:SumRecurrence} with the \texttt{MultiSum}-package. To be more precise,
in several cases we succeeded in finding a structure set $S$ with the corresponding degree
bounds for the polynomial coefficients, but we failed to determine the summand
recurrence~\eqref{non-kfree-Feynman} explicitly, since the underlying linear system was too large to solve.
For such situations, we dropped, e.g., the outermost summation quantifier, say $\sum_{\sigma_1=p_1}^{\infty}$ and searched for a recurrence in $\sigma_1$; in particular the variable $N$ was put in the base field $\set K$. In this simpler form, we succeeded in finding a recurrence. Next, we computed the initial values (in terms of $N$) by using another round of our method. With this input, Algorithm~\FLSR\ found an expansion with coefficients in terms of $F_{t}(N,\sigma),F_{t+1}(N,\sigma),\dots,F_{u}(N,\sigma)$. To this end, we applied the infinite sum

\vspace*{-0.3cm}

\begin{equation}\label{Equ:LastSum}
\sum_{\sigma_1=p_1}^{\infty}F_i(\sigma,N)
\end{equation}
to the coefficients $F_{i}(N,\sigma)$ and simplified these expressions further by the techniques described in~\cite{HYP2}. In various situations, it turned out that this hybrid technique was preferable to computing a pure recurrence in $N$ or just simplifying the expressions~\eqref{eq:A5} by using the methods given in~\cite{HYP2}.\\
\textsf{3. Asymptotic expansions for infinite expressions.} As mentioned in Remark~\ref{Remark:SplittingSums} we obtained also sums of the form~\eqref{physics-sums-general} which could be defined only by considering a truncated version of the infinite sums. For such cases we computed the coefficients $F_i(\sigma,N)$ as above and considered --instead of~\eqref{Equ:LastSum}-- the expressions
$\sum_{\sigma=0}^{a}F_i(\sigma,N)$
for large values $a$. To be more precise, we computed asymptotic expansions for all these sums and combined them to one asymptotic expansion in $a$. In this final form all the expressions canceled which were not defined when performing $a\to\infty$ and we ended up with the correct $F_i(N)$.\\
\textsf{4. Dealing with several infinite sums.} In all our computations only a single infinite sum arose. In principle, our method works also in the case when there are several such sums. However, in order to set up the recurrence in Section~\ref{Sec:FindRecurrences}, we need additional properties such as~\eqref{Equ:SumDefined} for the multivariate case. If such properties are not available, we propose two strategies: \textsf{4.1} Drop some (or all) of the infinite sums and proceed as explained in point 2 of our remark. \textsf{4.2} Set up the recurrence with formal sums and expand the sums on the right hand side: here one can either use the strategies as described in Step~4 of Section~\ref{Sec:Physic} (in particular, if asymptotic expansions have to be computed), or one can proceed with the method of this section whenever the sum is analytically well defined.
\end{remark}

\section{Conclusion}\label{Sec:Conclusion}

We presented a general framework that enables one to compute the first coefficients $F_i(N)$ of the Laurent expansion of a given Feynman parameter integral, whenever the $F_i(N)$ are expressible in terms of indefinite nested product-sum expressions. Namely, starting from such integrals, we described a symbolic approach to obtain a multi-sum representation over hypergeometric terms. Given this representation, we developed symbolic summation tools to extract these coefficients from its sum representation. In order to tackle this problem, Wegschaider's \texttt{MultiSum} package has been enhanced with Stan's package \texttt{FSum} that handles sums which do not satisfy finite support
conditions. Moreover, given a recurrence relation of the form~\eqref{Equ:ExpansionEqu} together with initial values, we used Schneider's recurrence solver that decides constructively, if the first coefficients of the formal Laurent series solution are expressible in terms of indefinite nested product-sum expressions.

In order to fit the input class of hypergeometric multi-sum packages, we split the sums at the price of possible divergencies. We overcame this situation by combining our new methods with other tools described, e.g., in~\cite{HYP2}; see Remark~\ref{Remark:MethodConclusion}. Further analysis of the introduced method should lead to a uniform approach that can handle in one stroke also solutions in terms of asymptotic expansions.

The described summation tools assisted in the task to compute
two- and simple three-loop diagrams, which occurred
in the calculation of the massive Wilson coefficients for deep-inelastic scattering; see~\cite{HYP2,HWILS1,HWILS2,HWILS3,sums-comp-description}. We are curious to see whether these new summation technologies find their application also in other fields of research.

\small


\end{document}